\pdfoutput=1
\documentclass[
  a4paper,
  USenglish,
  cleveref,
  autoref,
  numberwithinsect
]{lipics-v2021}
\usepackage{booktabs}
\usepackage{cite}
\usepackage{algorithm}
\usepackage[noend]{algpseudocode}
\AtBeginDocument{\hypersetup{pdfsubject={}}}
\hideLIPIcs
\nolinenumbers
\algrenewcommand{\algorithmicrequire}{\textbf{Input:}}
\crefname{algorithm}{Algorithm}{Algorithms}
\usepackage{physics}
\newcommand{\eps}{\varepsilon}
\newcommand{\E}{\mathbb{E}}
\newcommand{\Prob}{\mathbb{P}}
\newcommand{\Ind}{\mathbf{1}}
\newcommand{\RE}{\textnormal{\textsc{RandEdge}}}
\newcommand{\Deg}{\textnormal{\textsc{Degree}}}
\DeclareMathOperator{\Var}{Var}
\DeclareMathOperator{\Cov}{Cov}
\DeclareMathOperator{\KL}{D_{KL}}
\DeclareMathOperator{\TV}{d_{TV}}

\title{The Matching Number Characterizes the Query Complexity of Edge Counting and Degree Moment Estimation}
\titlerunning{Matching Number for Edge and Degree Moment Estimation}

\author{Chenglin Fan}{Seoul National University, Republic of Korea}{fanchenglin@gmail.com}{}{}
\author{Wonseok Shin}{Standigm Inc., Republic of Korea}{wonseok.shin.dev@gmail.com}{https://orcid.org/0000-0002-5311-2511}{}
\authorrunning{C. Fan and W. Shin}
\Copyright{Chenglin Fan and Wonseok Shin}


\relatedversion{}

\begin{document}

\maketitle

\begin{abstract}
Consider an unknown graph $G = (V,E)$ with vertex degrees $d_v$, accessed only by sampling uniformly random edges and querying the degrees of their endpoints.
We study the query complexity of estimating the number of edges $m$ and the degree moments $M_s = \sum_{v \in V} (d_v)^s$ for fixed integers $s \ge 2$.
In this model, [Beretta and T\v{e}tek, TALG 2024] estimate $m$ within relative error $\varepsilon$ using $O(\log n/\varepsilon^2+\sqrt n/\varepsilon)$ expected queries without knowing the number of vertices $n$, and [Bhattacharya and Pradhan, CIAC 2025] estimate $M_s$ using $O(n^{1-1/s}/\varepsilon^2)$ queries when $n$ is given.
We show that the query complexity of both problems is governed by the matching number $\nu$ of $G$ (the maximum number of pairwise disjoint edges) rather than the number of vertices $n$, and that the dependence on $\varepsilon$ for degree moments can be polynomially improved.
Together with matching lower bounds, this characterizes the worst case expected query complexity of both problems over graphs of bounded matching number.
Our algorithms rest on a simple observation: the endpoints of a maximum matching cover every edge, so a random endpoint of a random edge falls on this set of $2\nu$ vertices with probability at least $1/2$.

For every $\varepsilon\in(0,1]$, we give \emph{advice-free} algorithms that estimate $m$ and $M_s$ within relative error $\varepsilon$ with probability at least $2/3$, without knowledge of $n$, $\nu$, or other graph parameters:
\begin{itemize}
  \item For $m$, our algorithm uses $O(\min\{\nu/\varepsilon,\ 1/\varepsilon^2+\sqrt\nu/\varepsilon\})$ expected queries.
    Compared with the bound of Beretta and T\v{e}tek, this replaces $n$ with $\nu$ and removes the $\log n$ factor from the $1/\varepsilon^{2}$ term.
  \item For $M_s$, our algorithm uses $O(\min\{\nu^{1-1/s}/\varepsilon^{2/s},\ \nu^{1-1/(2s-1)}/\varepsilon^{2/(2s-1)}\})$ expected queries.
    Compared with the bound of Bhattacharya and Pradhan, this replaces $n$ with $\nu$, removes the need to know $n$, and polynomially improves the dependence on $\varepsilon$.
\end{itemize}
Finally, our lower bounds show that for every $k \ge 2$ and sufficiently small $\eps$, these query complexities are optimal up to constant factors over graphs with matching number at most $k$. 

\end{abstract}

\newpage

\section{Introduction}
Estimating graph parameters such as edge counts and higher degree moments from local observations is a central challenge in sublinear time algorithms.
Consider an unknown finite, simple, undirected graph $G = (V, E)$ with $n := \abs{V}$ vertices, $m := \abs{E} > 0$ edges, and vertex degrees $d_v$.
We define the total degree sum as $W := \sum_{v \in V} d_v = 2m$, and for a fixed integer $s \ge 2$, the $s$-th degree moment as $M_s := \sum_{v \in V} (d_v)^s$.
The query complexity of estimating these parameters depends substantially on the graph access model, structural assumptions, and whether parameters such as the vertex count $n$ are provided as advice.

While a rich body of literature studies sublinear graph algorithms based on uniform vertex sampling, often augmented by degree or neighbor queries~\cite{Feige2006,GoldreichRon2008,EdenRonSeshadhri2019}, sampling vertices uniformly at random is often unavailable.
In many realistic settings, the graph is accessed only through observations of edges and their endpoints.
For instance, the rows of a binary relation in a relational database correspond to the edges of a graph~\cite{AssadiKapralovKhanna2019}: drawing a uniformly random row corresponds to sampling a random edge, while an index supporting frequency queries provides the degree of an observed endpoint~\cite{AliakbarpourEtAl2018}.
This motivates studying sublinear graph algorithms based solely on sampling random edges and querying endpoint degrees.
Under this model, Beretta and T\v{e}tek~\cite{BerettaTetek2024} established that an estimate of $m$ within relative error $\eps \in (0, 1]$ can be obtained in $O(\sqrt{n} / \eps)$ queries when the vertex count $n$ is known, and in $O(\log n / \eps^2 + \sqrt{n}/\eps)$ expected queries without advice.

Beyond edge counts, estimating higher degree moments $M_s$ for $s \ge 2$ is also a fundamental problem in network analysis and database query optimization, where $M_2$ corresponds to self-join sizes~\cite{AliakbarpourEtAl2018}.
For this problem, Bhattacharya and Pradhan~\cite{BhattacharyaPradhan2025} recently showed that $M_s$ can be estimated within relative error $\eps$ in $O(n^{1-1/s}/\eps^2)$ queries when $n$ is known.

Despite recent advances, existing algorithms or their analyses still rely on the vertex count $n$, which often serves as a poor measure of instance difficulty; indeed, Beretta and T\v{e}tek~\cite{BerettaTetek2024} explicitly ask whether the complexity can be parameterized by quantities beyond the input size.
For graphs, this raises a natural question: which structural property of the graph governs the query complexity of estimating $m$ and $M_s$ from random edges, and can an algorithm achieve optimal complexity without any advice?

We show that this complexity is governed by the matching number $\nu := \nu(G)$, the maximum number of pairwise disjoint edges in $G$.
The endpoints of a maximum matching form a vertex cover of size $2\nu$. Because every edge has at least one endpoint in this cover, a random endpoint of a random edge falls on these $2\nu$ vertices with probability at least $1/2$.
This concentration gives query complexities parameterized by $\nu$ instead of $n$, yielding significantly stronger guarantees on structured instances.
In addition, for degree moments $M_s$ with a fixed integer $s \ge 2$, our algorithm polynomially improves the dependence on $\eps$.
Moreover, all our algorithms operate without knowledge of $n$, $m$, $\nu$, or any parameter bounds.
For every integer $k \ge 2$ and sufficiently small $\eps$, our algorithms and lower bounds characterize the worst case expected query complexity over graphs with $\nu(G) \le k$ up to constant factors depending only on $s$.

\subsection{Problem Definition and Query Model}
\label{sec:model}
We work with random edge sampling and degree queries at sampled endpoints, which we refer to as \emph{endpoint access}.
This access model corresponds to a graph realization of weighted sampling (elements from a weighted set are sampled with probability proportional to their weights, and each sample reveals the element's identifier and weight)~\cite{BerettaTetek2024,BhattacharyaPradhan2025} and was also studied for star counting by Aliakbarpour et al.~\cite{AliakbarpourEtAl2018}.
We formalize it through two unit cost oracles:
\begin{itemize}
  \item $\RE()$: Returns an edge $\{u,v\} \in E$ chosen uniformly at random with replacement, independently of all previous queries.
  \item $\Deg(v)$: Returns the degree $d_v$, where $v$ must be an endpoint returned by an earlier $\RE()$ call.
\end{itemize}
Returned vertices are revealed as opaque identifiers; an algorithm can test whether two endpoints coincide, but cannot sample a uniform vertex or enumerate $V$.

A \emph{fair endpoint sample} is an endpoint $X$ of the edge returned by one $\RE()$ call, chosen by a fair coin flip, and thus follows the degree proportional law
\[
  \Prob[X = v] = p_v := \frac{d_v}{2m} = \frac{d_v}{W} \qquad (v \in V).
\]
Since isolated vertices have $p_v = 0$, they are never sampled under endpoint access and contribute nothing to $m$, $M_s$, or $\nu$. We allow input graphs to have isolated vertices; our algorithms never encounter them, so isolated vertices do not affect the query complexity.
Because drawing a fair endpoint sample and querying its degree requires two queries, our access model simulates the weighted sampling access used in prior work with constant factor overhead.

Given an accuracy parameter $\eps \in (0, 1]$, we study the problem of estimating the number of edges $m$ and the $s$-th degree moment $M_s$.
For a target quantity $\theta$, an estimate $\hat{\theta}$ is a \emph{$(1 \pm \eps)$-estimate} of $\theta$ if $(1-\eps)\theta \le \hat{\theta} \le (1+\eps)\theta$.
An algorithm computes a $(1\pm\eps)$-estimate if it returns such an estimate with probability at least $2/3$,\footnote{The success probability can be boosted to $1-\delta$ by returning the median of the outputs of $O(\log(1/\delta))$ independent runs.} and is \emph{advice-free} if it is given only $\eps$ (and the integer $s \ge 2$ when estimating $M_s$) besides the oracles; in particular, it receives no prior knowledge of $n$, $m$, $M_s$, $\nu$, or bounds on any of them.

\subsection{Our Results}
We design advice-free algorithms whose query complexity is parameterized by the matching number $\nu$ and establish matching worst case lower bounds for estimating edge counts and degree moments under endpoint access.
\Cref{tab:endpoint-comparison} summarizes our main results.

\begin{table}[t]
\caption{Expected query complexity of obtaining $(1\pm\eps)$-estimates with success probability at least $2/3$ under endpoint access, for a fixed integer $s\ge2$. Lower bounds are worst case and assume sufficiently small $\varepsilon$. Here $n$ is the vertex count and $\nu$ is the matching number.}
\label{tab:endpoint-comparison}
\centering
\renewcommand{\arraystretch}{1.0}
\resizebox{\linewidth}{!}{%
\begin{tabular}{@{}l @{\hspace{3pt}} c @{\hspace{3pt}} c @{\hspace{6pt}} c @{\hspace{6pt}} c@{}}
\toprule
\textbf{Result} & \textbf{Target} & \textbf{Advice} & \textbf{Query complexity} & \textbf{Lower bound} \\
\midrule
Beretta and T\v{e}tek~\cite{BerettaTetek2024} & $m$ & $n$
  & $\displaystyle O\!\left(\sqrt{n}/\eps\right)$
  & $\displaystyle \Omega(\sqrt{n} / \eps)$\textsuperscript{$\dagger$} \\[4pt]
Beretta and T\v{e}tek~\cite{BerettaTetek2024} & $m$ & ---
  & $\displaystyle O\!\left(\log n / \eps^2 + \sqrt{n} / \eps\right)$
  & $\displaystyle \Omega\!\left(\sqrt{n} / \eps \right)$\textsuperscript{$\dagger$} \\[4pt]
\textbf{This work} & $m$ & ---
  & \begin{tabular}[c]{@{}c@{}}$\displaystyle O\!\left(\min\!\left\{\frac{\nu}{\eps},\,\frac1{\eps^2}+\frac{\sqrt\nu}{\eps}\right\}\right)$\\[2pt](\Cref{thm:intro-edge})\end{tabular}
  & \begin{tabular}[c]{@{}c@{}}$\displaystyle \Omega\!\left(\min\!\left\{\frac{k}{\eps},\,\frac1{\eps^2}+\frac{\sqrt k}{\eps}\right\}\right)\textsuperscript{$\ddagger$}$\\[2pt](\Cref{thm:intro-edge-lb})\end{tabular} \\
\midrule[1.2pt]
Bhattacharya and Pradhan~\cite{BhattacharyaPradhan2025} & $M_s$ & $n$
  & $\displaystyle O\!\left(n^{1-1/s}/\eps^2\right)$\textsuperscript{$\ast$}
  & ---\textsuperscript{$\ast\ast$} \\[4pt]
\textbf{This work} & $M_s$ & ---
  & \begin{tabular}[c]{@{}c@{}}$\displaystyle O\!\left(\min\!\left\{\frac{\nu^{1-1/s}}{\eps^{2/s}},\,\frac{\nu^{1-1/(2s-1)}}{\eps^{2/(2s-1)}}\right\}\right)$\\[2pt](\Cref{thm:intro-moment})\end{tabular}
  & \begin{tabular}[c]{@{}c@{}}$\displaystyle \Omega\!\left(\min\!\left\{\frac{k^{1-1/s}}{\eps^{2/s}},\,\frac{k^{1-1/(2s-1)}}{\eps^{2/(2s-1)}}\right\}\right)\textsuperscript{$\ddagger$}$\\[2pt](\Cref{thm:intro-moment-lb})\end{tabular} \\
\bottomrule
\end{tabular}}
\par\smallskip
\begin{minipage}{\linewidth}
\footnotesize
The prior algorithms transfer from weighted sampling to endpoint access; see \Cref{sec:proportional-transfer}.\\
\textsuperscript{$\dagger$}Holds for $\eps\ge2/\sqrt n$.
\quad
\textsuperscript{$\ddagger$}Over graphs with $\nu(G)\le k$, for integers $k\ge2$.\\
\textsuperscript{$\ast$}The stated instance sensitive bound (not shown) is not achieved by the proposed algorithm; see \Cref{sec:moment-instance-sensitive}.\\
\textsuperscript{$\ast\ast$}The lower bound does not hold uniformly over the stated parameter range; see \Cref{sec:bp25-lower-bound}.
\end{minipage}
\end{table}

\subparagraph*{Estimating the number of edges.}
Our first main result is an advice-free algorithm for estimating the edge count $m$.

\begin{theorem}
\label{thm:intro-edge}
Under the endpoint access model, for every $0 < \eps \le 1$, there exists an advice-free algorithm that terminates almost surely, returns a $(1\pm\eps)$-estimate of $m$ with probability at least $2/3$, and has expected query complexity
$
  O\!\left(\min\!\left\{
  \nu / \eps,
  1/\eps^2+\sqrt{\nu}/ \eps
  \right\}\right).
$
\end{theorem}
Compared with the $O(\log n / \eps^2 + \sqrt{n}/\eps)$ expected query complexity of Beretta and T\v{e}tek~\cite{BerettaTetek2024} without $n$, our bound replaces $n$ with the matching number $\nu \le n/2$ and removes the $\log n$ factor.
We note that while Beretta and T\v{e}tek achieve an $O(\sqrt n/\varepsilon)$ guarantee when $n$ is known, our bound instead depends on the matching number and requires no knowledge of $n$.

We show that this query complexity is optimal by establishing a matching worst case lower bound over graphs of bounded matching number.
For an integer $k \ge 2$, let $\mathcal{G}_{\le k}$ denote the family of finite simple graphs $G$ with $|E(G)| > 0$ and $\nu(G) \le k$.\footnote{
We restrict to $k \ge 2$ since if $\nu = 1$ is known, deleting isolated vertices leaves a star or a triangle, which can be identified with $O(1)$ queries.
Our algorithms remain valid for such inputs without advice.
}
\begin{theorem}
\label{thm:intro-edge-lb}
Under the endpoint access model, there exists a constant $\eps_0 > 0$ such that for every integer $k \ge 2$ and every $0 < \eps \le \eps_0$, any algorithm, even if given $k$ and $n$, that returns a $(1\pm\eps)$-estimate of $m$ with probability at least $2/3$ on every graph in $\mathcal{G}_{\le k}$ has worst case expected query complexity $\Omega(\min\{k / \eps, 1 / \eps^2 + \sqrt{k} / \eps\})$.
\end{theorem}
Together with \Cref{thm:intro-edge}, this establishes the optimal complexity over $\mathcal{G}_{\le k}$ up to constant factors: $\Theta(k/\eps)$ for $k\le\eps^{-1}$, $\Theta(\eps^{-2})$ for $\eps^{-1}\le k\le\eps^{-2}$, and $\Theta(\sqrt{k}/\eps)$ for $k\ge\eps^{-2}$.

\subparagraph*{Estimating degree moments.}
Our second main result is an advice-free algorithm for estimating the degree moments $M_s$.

\begin{theorem}
\label{thm:intro-moment}
Fix an integer $s \ge 2$. Under the endpoint access model, for every $0 < \eps \le 1$, there exists an advice-free algorithm that terminates almost surely, returns a $(1\pm\eps)$-estimate of $M_s$ with probability at least $2/3$, and has expected query complexity
\[
  O\!\left(\min\!\left\{
  \frac{\nu^{1-1/s}}{\eps^{2/s}},
  \frac{\nu^{1-1/(2s-1)}}{\eps^{2/(2s-1)}}
  \right\}\right),
\]
where the hidden constant depends only on $s$.
\end{theorem}
Bhattacharya and Pradhan obtain an $O(n^{1-1/s}/\eps^2)$ query complexity when $n$ is known. Our bound replaces $n$ with $\nu$ and polynomially improves the dependence on $\eps$, while requiring no graph parameters as advice.
Finally, we show that this query complexity is optimal in the worst case over all graphs of bounded matching number.
\begin{theorem}
\label{thm:intro-moment-lb}
Fix an integer $s \ge 2$. Under the endpoint access model, there exists a constant $\eps_0(s) > 0$ such that for every integer $k \ge 2$ and every $0 < \eps \le \eps_0(s)$, any algorithm, even if given $k, n$, and $m$, that returns a $(1\pm\eps)$-estimate of $M_s$ with probability at least $2/3$ on every graph in $\mathcal{G}_{\le k}$ has worst case expected query complexity
\[
  \Omega\!\left(\min\!\left\{
  \frac{k^{1-1/s}}{\eps^{2/s}},
  \frac{k^{1-1/(2s-1)}}{\eps^{2/(2s-1)}}
  \right\}\right),
\]
where the hidden constant depends only on $s$.
\end{theorem}
Together with \Cref{thm:intro-moment}, this establishes the optimal complexity over $\mathcal{G}_{\le k}$ up to constant factors: $\Theta(k^{1-1/s}/\eps^{2/s})$ for $k \ge \eps^{-2}$ and $\Theta(k^{1-1/(2s-1)}/\eps^{2/(2s-1)})$ for $k \le \eps^{-2}$.

\subsection{Technical Overview}

\subparagraph*{The central obstacle of the endpoint access model.}
Because the vertex count is unknown, the size of the graph must be inferred from repeated observations of the same vertex.
Such repetitions can be rare, since a fair endpoint sample falls on each vertex only in proportion to its degree ($p_v = d_v / (2m)$) while many low degree vertices may share the probability mass.
Our algorithms overcome this obstacle through the mass concentration observation noted above: a fair endpoint sample falls on the $2\nu$ endpoints of a maximum matching with probability at least $1/2$, so repetitions reliably occur among these vertices regardless of how many low degree vertices exist.

\subparagraph*{Efficient algorithms from mass concentration on matching endpoints.}
For edge counting (\Cref{sec:edge-upper}), the rate $O(\min\{\nu/\eps, 1/\eps^2 + \sqrt{\nu}/\eps\})$ arises from two complementary algorithms.
When $\nu$ is large, \Cref{alg:endpoint-capture} (\Cref{sec:edge_upper_1}) measures how frequently fair endpoint samples hit a set $S$ of distinct vertices formed from a batch of $q$ fair endpoint samples.
The mass concentration observation ensures that the vertices in $S$ carry expected total endpoint sampling probability $\Omega(q/(\nu+q))$. 
By testing this probability while doubling $q$, the algorithm ensures that waiting for $\Theta(\eps^{-2})$ hits in $S$ costs only $O(1/\eps^2+\sqrt{\nu}/\eps)$ expected queries.
When $\nu$ is small, \Cref{alg:sampled-cover} (\Cref{sec:edge_upper_2}) greedily collects a vertex set of size $O(\nu)$ that covers almost all edges.
It then estimates the number of edges with both endpoints in this set via repeated edge observations, recovering $m$ in $O(\nu/\eps)$ queries.
Interleaving both algorithms query by query asymptotically achieves the minimum of the two bounds without knowing $\nu$.

For degree moments, high degree vertices make large contributions $(d_v)^s$ to $M_s$, so repeated observations of the same vertex cause severe variance.
We avoid this by using a Horvitz--Thompson estimator that counts each observed vertex only once, weighted inversely by its probability of being observed (\Cref{sec:moment-upper}).
This causes the relative variance to decay as higher powers of the sample size, which avoids the $\eps^{-2}$ cost of sample averaging to achieve a polynomial improvement in $\eps$.
\Cref{sec:moment-upper} further shows that this variance is governed by the scale $\Lambda_s := W / M_s^{1/s}$, which our mass concentration observation bounds by $\Lambda_s \le 2(2\nu)^{1-1/s}$ on every graph, achieving the bounds of \Cref{thm:intro-moment}.

\subparagraph*{Lower bounds from structural bottlenecks.}
To establish our lower bounds (\Cref{sec:edge-lower,sec:moment-lower}), we construct pairs of graph distributions that are hard to distinguish.
For edge counting, we first construct graphs where the edges that differ between the two distributions are diluted among leaves attached to $k$ vertices of equal degree, requiring $\Omega(\min\{k/\eps, 1/\eps^2\})$ queries (\Cref{sec:lb_edge_1}).
When $k \ge \eps^{-2}$, we adapt the construction of Beretta and T\v{e}tek to graphs, showing that counting isolated edges must rely on repeated edge observations as all observed degrees are one and thus requires $\Omega(\sqrt{k}/\eps)$ queries (\Cref{sec:lb_edge_2}).
For degree moments with sufficiently large $k\eps^2$, our first construction uses $\Theta(\eps^{-2})$ stars among isolated edges to hide a small change in the number of stars, requiring $\Omega(k^{1-1/s}/\eps^{2/s})$ queries (\Cref{sec:moment-stars}).
In the complementary regime, we construct graphs with at most $k$ center vertices whose incident edges can share auxiliary endpoints while center degrees stay fixed, requiring $\Omega(k^{1-1/(2s-1)}/\eps^{2/(2s-1)})$ queries (\Cref{sec:moment-paired}).

\subsection{Further Related Work}
\subparagraph*{Vertex query models.}
Given the vertex count $n$, sublinear algorithms estimate the average degree within factor $2+\eps$ using uniform vertex and degree queries~\cite{Feige2006}, and achieve $(1\pm\eps)$-estimates for average degrees and higher degree moments when neighbor queries are also available~\cite{GoldreichRon2008,EdenRonSeshadhri2019}.
Average degree estimation has also been investigated under arboricity bounds~\cite{EdenRonSeshadhri2019,EdenRubinfeldVasilyan2026,EdenSeshadhri2026} or without knowing $n$ under promises on minimum average degree or nonisolated vertices~\cite{Chanda2025,BerettaChakrabartySeshadhri2026}.
Nonadaptive edge estimation has been characterized when $V=[n]$ is known~\cite{BishnuEtAl2025Nonadaptive}, and edge counting has also been studied under hash-ordered neighbor access~\cite{TetekThorup2022}.
Related work also studies generating uniform or approximately uniform random edges using vertex and neighbor queries~\cite{EdenRosenbaum2018Sampling,EdenRonRosenbaum2019Sampling,EdenNarayananTetek2023}.

\subparagraph*{Subgraph counting.}
A broader line of work investigates estimating counts of stars~\cite{GonenRonShavitt2011,AliakbarpourEtAl2018}, cliques and triangles~\cite{EdenLeviRonSeshadhri2017,EdenRonSeshadhri2020Cliques}, and general subgraphs~\cite{AssadiKapralovKhanna2019}.
Among these, counting $s$-stars $S_s := \sum_v \binom{d_v}{s}$ is closely tied to degree moments.
Given $n$ and $m$, Aliakbarpour et al.~\cite{AliakbarpourEtAl2018} estimate $S_s$ using $O(m\log\log n / (\eps^2 S_s^{1/s}))$ random edge and degree queries, where the scale $m/S_s^{1/s}$ governing their variance corresponds to $\Lambda_s = W/M_s^{1/s}$ for degree moments.
However, their algorithm requires $n$ to search over candidate scales and its query complexity grows as $\eps^{-2}$.
By contrast, our estimator is advice-free, has no $\log\log n$ overhead, and improves the dependence on $\eps$ to $\eps^{-2/s}$ and $\eps^{-2/(2s-1)}$.

\subparagraph*{Weighted sampling and moment estimation.}
Sum and moment estimation have also been widely studied under weighted sampling~\cite{MotwaniPanigrahyXu2007,BerettaTetek2024,BhattacharyaPradhan2025}.
Because our moment estimator only uses sampled endpoints and their degrees, \Cref{thm:intro-moment} transfers to arbitrary weighted sets of $n$ elements with $n$ in place of $\nu$ (\Cref{sec:proportional-transfer}).
The resulting sample complexity is inconsistent with the $\Omega(n^{1-1/s}/\eps^2)$ lower bound stated in~\cite{BhattacharyaPradhan2025} at constant failure probability: \Cref{sec:bp25-lower-bound} shows that the lower bound does not hold uniformly over the stated parameter range and examines the construction used in its proof.

\section{Algorithms for Edge Count Estimation}
\label{sec:edge-upper}
We give two complementary algorithms for estimating the number of edges $m$.
Recall that a fair endpoint sample $X$ satisfies $\Prob[X=v] = p_v = d_v / (2m)$.
For any subset $A \subseteq V$, we define its total degree $D(A)$ and its endpoint mass $p_A$ as
\[
  D(A) := \sum_{v\in A} d_v,
  \qquad
  p_A := \Prob[X \in A] = \sum_{v\in A} p_v = \frac{D(A)}{2m}.
\]
The two algorithms utilize $D(A)$ through distinct strategies: \Cref{alg:endpoint-capture} obtains $m$ by estimating the endpoint mass $p_S$ of a vertex set $S$ formed from fair endpoint samples, whereas \Cref{alg:sampled-cover} reconstructs $m$ from a vertex set covering almost all edges and an estimate of the number of edges with both endpoints in that set.

\subsection{Estimating Edge Count via Endpoint Mass}
\label{sec:edge_upper_1}
\begin{algorithm}[t]
\caption{Estimating $m$ via endpoint mass}
\label{alg:endpoint-capture}
\begin{algorithmic}[1]
\Require Endpoint access to $G$; accuracy $0<\eps\le1$.
\State $R\gets\lceil64\eps^{-2}\rceil$.
\For{$j=0,1,2,\ldots$}
  \State $q\gets2^j$.
  \State Draw $q$ fair endpoint samples; let $S$ be the set of distinct vertices obtained.
  \State Draw $q$ fair endpoint samples; let $H$ count those in $S$, with multiplicity.
  \If{$H<R$}
    \State \textbf{continue};
  \EndIf
  \State Query $\Deg(v)$ once for every $v\in S$ and compute $D(S)$.
  \State Draw fair endpoint samples until $R$ samples lie in $S$; let $T$ count all trials.
  \State \Return $D(S)T/(2R)$.
\EndFor
\end{algorithmic}
\end{algorithm}

Let $S$ be the set of distinct vertices in a batch of fair endpoint samples.
Since $2m = D(S) / p_S$ holds for any $S \subseteq V$ with $p_S > 0$ and $D(S)$ can be evaluated using $\abs{S}$ degree queries, estimating $m$ reduces to estimating $1/p_S$.
A standard approach to obtain a $(1\pm\eps)$-estimate of $1/p_S$ with constant success probability is to draw independent fair endpoint samples until observing $R = \Theta(\eps^{-2})$ hits in $S$.
If $T$ denotes the required number of trials, then $\E[T \mid S] = R / p_S$, yielding the unbiased estimator $\widehat{m} := D(S)T / (2R)$.

However, if $p_S$ is small, the expected number of trials $R / p_S$ may become prohibitively large.
To prevent this, \Cref{alg:endpoint-capture} tests whether $S$ has sufficient endpoint mass before querying degrees.
Iterating over $q = 2^j$ for $j = 0, 1, 2, \dots$, it forms $S$ from $q$ fair endpoint samples and counts hits $H$ in a second batch of $q$ samples.
If $H < R$ for $R := \lceil 64\eps^{-2} \rceil$, the endpoint mass is deemed insufficient and the algorithm doubles $q$ to resample $S$.
Otherwise, it queries the degrees of the vertices in $S$ to compute $D(S)$ and draws fair endpoint samples until observing $R$ hits in $S$, returning $\widehat{m} = D(S) T / (2R)$.

\begin{theorem}
\label{thm:endpoint-capture}
Under the endpoint access model, for every $0<\eps\le1$, \Cref{alg:endpoint-capture} is advice-free, terminates almost surely, and returns a $(1\pm\eps)$-estimate of $m$ with probability at least $63/64$.
Its expected query complexity is $O(1/\eps^2 + \sqrt{\nu} / \eps)$, and the returned estimate $\widehat m$ satisfies $\E[(\widehat m / m)^2]=O(1)$ and $\E[(m / \widehat m)^2]=O(1)$.
\end{theorem}
The two ratio moment bounds are used in the proof of \Cref{lem:moment-search}.

The matching number enters through the endpoint mass $p_S$ of the set $S$ formed from the first $q$ samples.
As at least half of the endpoint mass lies on the $2\nu$ endpoints of a maximum matching, $\E[p_S]$ can be bounded below in terms of $q$ and $\nu$ (\Cref{lem:captured-endpoint-mass}).
Since $\E[H]=q\E[p_S]$, this controls how large $q$ must be before we expect enough hits.
Specifically, \Cref{lem:capture-hit-threshold} establishes that once $q$ exceeds a sufficiently large constant multiple of $R+\sqrt{\nu R}$, each reached iteration continues with probability at most $1/4$.
Thus the probability of reaching larger values of $q$ decays geometrically beyond this threshold.

It remains to bound the number of trials $T$ required to observe $R$ hits in $S$.
\Cref{lem:capture-fixed-scale-cost} shows that, for a fixed set $S$ and setting $T := 0$ if $H < R$,
\[
  \E[T \mid S] = \Prob[H \ge R \mid S] \cdot \frac{R}{p_S}
  \le \frac{q p_S}{R} \cdot \frac{R}{p_S}
  = q,
\]
where the inequality follows from Markov's inequality using $\E[H \mid S] = q p_S$.
The two initial batches cost $2q$ random edge queries, the degree queries cost at most $\abs{S}\le q$, and the final trials contribute $\E[T \mid S] \le q$, bounding the expected query cost of each reached iteration by $4q$.
Summing over all iterations yields the overall expected query complexity of $O(R + \sqrt{\nu R}) = O(1/\eps^2 + \sqrt{\nu}/\eps)$.
The complete proof appears in \Cref{sec:endpoint-capture-detail}.

\subsection{Estimating Edge Count via Covering Almost All Edges}
\label{sec:edge_upper_2}
\begin{algorithm}[t]
\caption{Estimating $m$ by covering almost all edges}
\label{alg:sampled-cover}
\begin{algorithmic}[1]
\Require Endpoint access to $G$; accuracy $0<\eps\le1$.
\For{$j=0,1,2,\ldots$}
  \State $b\gets2^j$, $q_b\gets\lceil16b/\eps\rceil$, $M_b\gets\varnothing$.
  \For{$i=1,\ldots,q_b$}
    \State $\{u,v\}\gets\RE()$.
    \If{$u\notin V(M_b)$ and $v\notin V(M_b)$}
      \State $M_b\gets M_b\cup\{\{u,v\}\}$.
    \EndIf
  \EndFor
  \If{$|M_b|<b$}
    \State $C\gets V(M_b)$ and \textbf{break};
  \EndIf
\EndFor
\State Query $\Deg(v)$ once for every $v\in C$ and compute $D(C)$.
\State Let $\widehat h$ be the output of \Cref{alg:internal-edges} on $(C,b,\eps)$.
\State \Return $\max\{0,D(C)-\widehat h\}$.
\end{algorithmic}
\end{algorithm}

With respect to any vertex set $C \subseteq V$, every edge in $G$ is either internal (both endpoints in $C$), crossing (exactly one endpoint in $C$), or uncovered (neither endpoint in $C$).
Let $h(C) := |E(G[C])|$ and $r(C) := |E(G[V \setminus C])|$ denote the numbers of internal and uncovered edges, respectively.
Because the total degree $D(C) = \sum_{v \in C} d_v$ counts internal edges twice, crossing edges once, and uncovered edges not at all, the edge count decomposes as $m = D(C) - h(C) + r(C)$.
Therefore, if we construct a set $C$ with $r(C) \le \eps m / 2$, compute $D(C)$ using $|C|$ degree queries, and estimate $h(C)$ by $\widehat h$ to additive error $\eps m / 2$, returning $\widehat m := \max\{0, D(C) - \widehat h\}$ achieves the desired accuracy.

\subparagraph*{Constructing $C$ by greedy matching.}
\Cref{alg:sampled-cover} constructs such a set $C$ without knowing $n$ or $\nu$ by greedily building a matching from random edges.
For a matching $M$, let $V(M)$ denote its set of endpoints.
Iterating over $b = 2^j$ for $j = 0, 1, 2, \dots$, the algorithm initializes an empty matching $M_b$, draws $q_b := \lceil 16b / \eps \rceil$ random edges, and inserts each sampled edge whose endpoints are both currently unmatched.
If $|M_b| < b$, the loop ends with $C := V(M_b)$; otherwise, it doubles $b$ and repeats.

\Cref{lem:sampled-approx-cover} proves that the set $C$ satisfies $r(C) \le \eps m / 2$ with probability at least $11/12$.
As long as more than $\eps m / 2$ edges remain uncovered, each draw enlarges the matching with probability at least $\eps / 2$.
Consequently, coupling the number of insertions with $\operatorname{Bin}(q_b,\eps/2)$ bounds the probability of stopping at scale $b$ with $r(C) > \eps m / 2$ by $e^{-3b}$.
Summing over all scales bounds the total failure probability by $\sum_{j \ge 0} e^{-3 \cdot 2^j} < 1/12$.
Because no matching has more than $\nu$ edges, the loop terminates at or before the first power of two $b_\star > \nu$, ensuring $b \le 2\nu$ and $|C| < 2b \le 4\nu$.
Drawing the samples across reached scales and querying degrees in $C$ uses $O(\nu/\eps)$ queries on every execution.

\subparagraph*{Estimating the internal edge count.}
\begin{algorithm}[t]
\caption{Estimating the internal edge count}
\label{alg:internal-edges}
\begin{algorithmic}[1]
\Require Endpoint access to $G$; a vertex set $C$, $b\ge1$ with $|C|<2b$, and $0<\eps\le1$.
\Statex An edge is internal if both endpoints lie in $C$. Fix a sufficiently large constant $K$.
\State $t\gets\lceil Kb/\eps\rceil$.
\State Draw $t$ independent random edges and retain the internal samples $e_1,\ldots,e_N$, including repetitions.
\If{$N<t\eps/4$}
  \State \Return $0$.
\EndIf
\State $Y\gets\binom{N}{2}^{-1}\sum_{1\le i<j\le N}\Ind[e_i=e_j]$.
\State \Return $0$ if $Y=0$, and $1/Y$ otherwise.
\end{algorithmic}
\end{algorithm}

Once $C$ is obtained, we estimate the internal edge count $h := h(C)$ with a collision estimator along the lines of~\cite{BerettaTetek2024}.
\Cref{alg:internal-edges} draws $t := \lceil Kb/\eps \rceil$ independent random edges for a sufficiently large constant $K$, and retains the $N$ internal samples, including repetitions.
When $h > 0$, the $N$ internal samples are independent and uniformly distributed over the $h$ internal edges.
Consequently, each pair of internal samples coincides with probability $1/h$, so when $N \ge 2$, the fraction of colliding pairs $Y$ has expectation $\E[Y \mid N] = 1/h$.
By the birthday paradox, observing $\Theta(\sqrt{h})$ internal samples suffices for constant relative accuracy.
Because $|C| < 2b$, the internal edge count satisfies $h \le \binom{|C|}{2} < 2b^2$, so $O(b)$ internal samples are sufficient.

When internal edges are more frequent, drawing $t$ random edges yields more internal samples, providing the higher accuracy needed for additive error $\eps m / 2$.
However, if internal edges are rare, $t$ random edges may yield too few internal samples to observe collisions.
Because each sampled edge is internal to $C$ with probability $h/m$, having $h \ge \eps m / 2$ ensures $\E[N] \ge t\eps/2$ and consequently $N \ge t\eps/4$ with high constant probability.
The algorithm therefore returns zero when $N < t\eps/4$, which is unlikely unless $h < \eps m / 2$, and returns $\widehat h = 1/Y$ (or zero if $Y = 0$) otherwise.
\Cref{lem:internal-correction} shows that this estimator satisfies $|\widehat h - h(C)| \le \eps m / 2$ with probability at least $11/12$ using $O(b/\eps)$ queries on every execution.

\begin{theorem}
\label{thm:sampled-cover}
Under the endpoint access model, for every $0<\eps\le1$, \Cref{alg:sampled-cover} is advice-free, uses $O(\nu/\eps)$ queries on every execution, and returns a $(1\pm\eps)$-estimate of $m$ with probability at least $5/6$.
\end{theorem}
The complete proof appears in \Cref{sec:sampled-cover-detail}.
We now combine \Cref{alg:endpoint-capture,alg:sampled-cover} to prove \Cref{thm:intro-edge}.
\begin{proof}[Proof of \Cref{thm:intro-edge}]
Run \Cref{alg:endpoint-capture,alg:sampled-cover} with accuracy parameter $\eps$, alternating their oracle queries, and return the first estimate produced.
Let $Q_1$ and $Q_2$ respectively denote the total numbers of queries made by the two algorithms when run independently to completion.
Interleaving their execution ensures that the total query count $Q$ satisfies $Q \le 2 \min\{Q_1, Q_2\}$, which immediately implies $\E[Q] \le 2 \min\{\E[Q_1], \E[Q_2]\}$.
Substituting the bounds from \Cref{thm:endpoint-capture,thm:sampled-cover} gives the claimed expected query complexity.

By a union bound, the probability that either algorithm fails to return a $(1\pm\eps)$-estimate is at most $1/3$ since each algorithm fails with probability at most $1/6$ (\Cref{thm:endpoint-capture,thm:sampled-cover}).
Thus, whichever algorithm terminates first returns a $(1 \pm \eps)$-estimate with probability at least $2/3$.
Finally, the combined procedure requires no advice and terminates in $O(\nu/\eps)$ queries on every execution as $Q \le 2Q_2$.
\end{proof}

\section{Lower Bounds for Edge Count Estimation}
\label{sec:edge-lower}
\subparagraph*{Lower bound strategy.}
Our lower bounds (\Cref{thm:intro-edge-lb,thm:intro-moment-lb}) are based on the standard reduction from estimation to hypothesis testing (Le Cam's method; see Yu~\cite{Yu1997}).
For a target graph parameter $\theta(G)$, we construct two distributions $\mathsf H_0,\mathsf H_1$ over simple graphs with $\nu(G) \le k$, alongside sets $I_0, I_1 \subset \mathbb{R}_{>0}$ satisfying $\Prob_{G \sim \mathsf{H}_i}[\theta(G) \in I_i] \ge 99/100$ for $i \in \{0, 1\}$.
These sets are $(1\pm\eps)$-separated: that is, no value is a $(1\pm\eps)$-estimate of both a value in $I_0$ and a value in $I_1$.
Consequently, any algorithm returning a $(1\pm\eps)$-estimate of $\theta(G)$ can identify whether $G$ was sampled from $\mathsf{H}_0$ or $\mathsf{H}_1$.

To bound the query complexity required to distinguish $\mathsf{H}_0$ from $\mathsf{H}_1$ under endpoint access, each construction considers a stronger \emph{augmented oracle}, whose responses to $q$ queries suffice to simulate any endpoint access algorithm making at most $q$ queries.
Let $\mathsf{P}_i$ denote the joint distribution of these responses under $\mathsf{H}_i$.
By the data processing inequality, the probability that any $q$-query algorithm correctly identifies whether $G\sim\mathsf H_0$ or $G\sim\mathsf H_1$ under a uniform prior is at most $(1 + \TV(\mathsf{P}_0, \mathsf{P}_1))/2$, where $\TV(\cdot, \cdot)$ denotes total variation distance.

In each construction, we show $\TV(\mathsf{P}_0, \mathsf{P}_1) \le 1/4$ for a chosen number of queries $q$.
If an algorithm succeeding with probability $2/3$ has worst case expected query complexity $o(q)$ on the union of the supports of $\mathsf H_0$ and $\mathsf H_1$, Markov's inequality implies that truncating execution at $q$ queries yields a test with success probability exceeding $(1 + 1/4)/2 = 5/8$, a contradiction.
The formal reduction appears in \Cref{sec:lower_bound_formal}.

We prove \Cref{thm:intro-edge-lb} using two lower bound constructions for two regimes of $k$:
\begin{itemize}
  \item When $k \le \eps^{-2}$, the \emph{degree-preserving edge splits} construction (\Cref{sec:lb_edge_1}) establishes an expected query lower bound of $\Omega(\min\{k/\eps, 1 / \eps^2\})$.
  \item When $k \ge \eps^{-2}$, the \emph{random matching subset} construction (\Cref{sec:lb_edge_2}) establishes an expected query lower bound of $\Omega(\sqrt{k}/\eps)$.
\end{itemize}
Together, these constructions establish the lower bound $\Omega(\min\{k/\eps, 1/\eps^2 + \sqrt{k}/\eps\})$ for all $k \ge 2$ and sufficiently small $\eps > 0$.
Throughout this section, the augmented oracle returns a uniformly random edge together with both endpoint degrees, which suffices because degree queries are restricted to endpoints of previously returned edges.

\subsection{Degree-Preserving Edge Splits}
\label{sec:lb_edge_1}
We first establish the lower bound $\Omega(\min\{k/\eps, 1/\eps^2\})$.
For simplicity, we assume that $k \ge k_0$ is a multiple of four for a constant $k_0$ fixed below; removing this assumption is discussed in \Cref{sec:lower_bound_edge_choice}.

\subparagraph*{Construction of degree-preserving edge splits.}
Let $L := \lceil 1/\eps \rceil$. We take a vertex set $U$ of $k$ vertices partitioned into two equal sets of size $k/2$, and place a $k/4$-regular bipartite graph $B$ between them, so $|E(B)| = k^2/8$.
To each vertex $u \in U$, we attach $L$ leaf vertices of its own; these $kL$ leaf edges are present in every graph.

For each edge $e = \{u,v\} \in E(B)$, we introduce two auxiliary vertices $x_e, y_e$ and form a \emph{gadget} on $\{u, v, x_e, y_e\}$ by selecting one of two local configurations:
in the \emph{single} configuration, the edge $\{u,v\}$ is present while $x_e, y_e$ remain isolated;
in the \emph{split} configuration, the two edges $\{u,x_e\}$ and $\{v,y_e\}$ are present while $\{u,v\}$ is absent.
We refer to the edges present in these configurations as \emph{gadget} edges, distinguishing them from the $kL$ leaf edges.
Both configurations contribute one incident edge to each of $u$ and $v$, so every vertex $u \in U$ has fixed degree $d_u = L + k/4$ regardless of gadget choices.
We define two graph distributions as follows:
\begin{description}
    \item[$\mathsf{H}_0$] Independently assign each $e \in E(B)$ to be \emph{single} or \emph{split} with probability $1/2$ each.
    \item[$\mathsf{H}_1$] Bias the probability of \emph{single} to $1/2 + \beta$, where $\beta:=C_0\eps(L+k)/k$ for a sufficiently large constant $C_0$.
\end{description}
With $C_0$ fixed, let $k_0$ and $\eps_0$ be constants such that $\beta \le 1/4$ whenever $k \ge k_0$ and $0 < \eps \le \eps_0$.
Because $U$ is a vertex cover of size $k$ and selecting one leaf edge at each vertex in $U$ gives a matching of size $k$, every graph in either support has matching number $\nu(G) = k$.
While the expected edge counts differ by $\E_{\mathsf{H}_0}[m] - \E_{\mathsf{H}_1}[m] = \beta k^2 / 8$, the standard deviation of $m$ under either hypothesis is $O(k)$.
By the choice of $C_0$, \Cref{lem:switch-separation} gives $(1\pm\eps)$-separated sets $I_0,I_1$ of edge counts such that $\Prob_{\mathsf H_i}[m\in I_i]\ge99/100$ for each $i\in\{0,1\}$.

\subparagraph*{Indistinguishability of degree-preserving edge splits.}
The degrees of all endpoints are determined by whether the sampled edge is a leaf edge, a single gadget edge, or a split gadget edge, as vertices in $U$ have degree $L + k/4$ while all other nonisolated vertices have degree $1$.
Thus any distinguishing information must arise from the sampled edge identities.

The indistinguishability follows from gadget edges being diluted by the $kL$ leaf edges.
A uniform random edge query hits a gadget with probability $\Theta(k/(L+k))$.
For a fixed gadget, the distributions of its configuration under the two hypotheses have KL divergence $O(\beta^2)$.
Because the probability of sampling each individual edge depends on the total number of edges $m$, \Cref{sec:lower_bound_edge_choice} analyzes rejection sampling from a fixed candidate universe and proves that the total KL divergence of $q$ responses is $O(q\beta^2 k/(L+k))$.
By Pinsker's inequality, this gives $\TV(\mathsf P_0, \mathsf P_1) \le 1/4$ for all $1 \le q \le c(L+k)/(\beta^2 k)$ for a sufficiently small constant $c > 0$, yielding the worst case expected query lower bound $\Omega((L+k)/(\beta^2 k)) = \Omega(\min\{k/\eps,1/\eps^2\})$.
The complete proof appears in \Cref{sec:lower_bound_edge_choice}.

\subsection{Random Matching Subset}
\label{sec:lb_edge_2}
To complete the proof of \Cref{thm:intro-edge-lb}, we establish the complementary lower bound $\Omega(\sqrt{k}/\eps)$ in the regime $k \ge \eps^{-2}$.

\subparagraph*{Construction of random matching subset.}
We adapt the weighted sampling hard family of Beretta and T\v{e}tek~\cite{BerettaTetek2024} to graphs by realizing each unit weight element as an isolated edge.
Let $M$ be a matching of $k$ isolated edges on $2k$ vertices. We define two graph distributions as follows:
\begin{description}
    \item[$\mathsf{H}_0$] A point mass on the full matching $G = M$.
    \item[$\mathsf{H}_1$] The subgraph of $M$ formed by a uniformly random $j$-subset of its edges, where $j := \lfloor(1 - 3\eps)k\rfloor$, leaving the remaining $2(k - j)$ vertices isolated.
\end{description}
Every graph in either support has matching number at most $k$, and every nonisolated vertex has degree $1$.
Since $(1+\eps)j \le (1+\eps)(1-3\eps)k < (1-\eps)k$, the edge counts $k$ and $j$ are $(1\pm\eps)$-separated.

\subparagraph*{Indistinguishability of random matching subset.}
Because every nonisolated vertex has degree $1$ under both hypotheses, endpoint degrees provide no distinguishing information.
\Cref{lem:hidden-support-likelihood} proves that the likelihood ratio of any sample sequence depends solely on the number of distinct observed edges, meaning that distinguishing the two distributions rests on the number of draws that repeat an earlier edge.

The probability of repeating an earlier edge after observing $r$ distinct edges is $r/k$ and $r/j$ under $\mathsf{H}_0$ and $\mathsf{H}_1$, respectively. 
Moreover, under both hypotheses, a repeated edge is uniform among the observed edges and a new edge is uniform among the unobserved edges.
Each step $t$ therefore reduces to distinguishing $\operatorname{Ber}(r/j)$ from $\operatorname{Ber}(r/k)$, which contributes KL divergence of $O(\eps^2 r / k) = O(\eps^2 t / k)$.
Summing over all steps $t < q$ gives a total divergence of $O(\eps^2 q^2 / k)$. 
By Pinsker's inequality, this gives $\TV(\mathsf P_0, \mathsf P_1) \le 1/4$ for all $1 \le q \le c\sqrt{k}/\eps$ for a sufficiently small constant $c > 0$, yielding the worst case expected query lower bound $\Omega(\sqrt{k}/\eps)$.
The complete proof appears in \Cref{sec:lower_bound_random_subset_matching}.

\section{Algorithm for Degree Moment Estimation}
\label{sec:moment-upper}
We now turn to estimating the degree moments $M_s = \sum_{v \in V} (d_v)^s$ for integers $s \ge 2$.
For the analysis in this section and \Cref{sec:moment-detail}, we may ignore isolated vertices as removing them preserves the sampling process, $m$, $M_s$, and $\nu$.
Unlike $m$, higher moments are dominated by high degree vertices, whose large contributions $(d_v)^s$ to $M_s$ cause severe variance under repeated sampling.
To control this variance, our algorithm estimates $M_s$ using an inverse-probability weighted Horvitz--Thompson estimator~\cite{HorvitzThompson1952} evaluated over the \emph{distinct} observed vertices.

\subsection{The Moment Estimator}
We first consider an ideal setting where the total degree sum $W = 2m$ is known.
Given an expected sample size $q > 0$, we draw $N \sim \operatorname{Pois}(q)$, take $N$ fair endpoint samples, and query the degree of each distinct observed vertex.
Let $S \subseteq V$ denote the set of distinct observed vertices.
Under Poissonization, each vertex $v \in V$ is observed independently with inclusion probability $\pi_v(q) := \Prob[v \in S] = 1 - e^{-q d_v / W}$.
We weight each observed vertex by its inverse inclusion probability, defining the Horvitz--Thompson estimator
\[
  Z_q(W) := \sum_{v \in S} \frac{(d_v)^s}{\pi_v(q)},
\]
where an empty sum evaluates to zero.
Because drawing $N$ fair endpoint samples and querying degrees of distinct vertices requires at most $2N$ queries, the expected query complexity is at most $2\E[N] = 2q$.

Counting each observed vertex only once controls the estimator's variance: an observed vertex $v$ contributes $(d_v)^s/\pi_v(q)$ however many times it is sampled.
Since a high degree vertex $v$ with large $q d_v / W$ is observed with probability $\pi_v(q)$ close to $1$, its contribution is nearly deterministic.
By contrast, weighting every sample of $v$ separately would leave a large variance of order $(d_v)^{2s} (q d_v / W)^{-1}$ for its contribution.
Poissonization also makes the events $\{v \in S\}$ mutually independent, yielding (\Cref{lem:moment-ht-props})
\[
  \E[Z_q(W)] = M_s, \qquad \Var[Z_q(W)] = \sum_{v \in V} \frac{(d_v)^{2s}}{e^{q d_v / W} - 1}.
\]

Define $\Lambda_s := W / M_s^{1/s}$.
For every integer $a \ge 1$ and $x \ge 0$, the exponential series yields $e^x - 1 \ge x^a / a!$.
Applying this with $x = q d_v / W$ for $a = s$ and $a = 2s - 1$ gives $\Var[Z_q(W)] \le s! (W/q)^s M_s$ and $\Var[Z_q(W)] \le (2s-1)! W^{2s} / q^{2s-1}$, respectively.
Dividing by $M_s^2$ yields
\[
  R_s(q) := \frac{\Var[Z_q(W)]}{M_s^2} = O\left(\min\!\left\{ \left(\frac{\Lambda_s}{q}\right)^s, \, \frac{\Lambda_s^{2s}}{q^{2s-1}} \right\}\right).
\]
By Chebyshev's inequality, taking $q = \Theta\left(\min\left\{\Lambda_s / \eps^{2/s}, \, \Lambda_s^{2s/(2s-1)} / \eps^{2/(2s-1)}\right\}\right)$ suffices to obtain a $(1\pm\eps)$-estimate with constant probability.

The matching number $\nu$ governs $\Lambda_s$ because the endpoints of a maximum matching form a vertex cover of size $2\nu$ containing at least half of the total degree sum.
By H\"older's inequality, this implies $\Lambda_s \le 2(2\nu)^{1-1/s}$ (\Cref{lem:moment-ht-relvar}), achieving the bounds of \Cref{thm:intro-moment}.
However, evaluating the denominators $\pi_v(q)$ requires $W$, and choosing $q$ requires $\Lambda_s$, neither of which is known in advance.

\subsection{Choosing the Parameters Without Advice}
To remove the dependence on $W$, we evaluate the same estimator with an arbitrary estimate $U > 0$ in place of $W$:
\begin{equation}
  Z_q(U) := \sum_{v \in S} \frac{(d_v)^s}{\pi_v(q, U)}, \qquad \pi_v(q, U) := 1 - e^{-q d_v / U}.
  \label{eq:moment-estimated-total-degree}
\end{equation}
Given $q$ and $U$, the algorithm draws $N \sim \operatorname{Pois}(q)$ and takes $N$ independent fair endpoint samples.
Replacing $W$ with $U$ alters only the denominators while each vertex is still observed with probability $\pi_v(q)$ as the samples do not depend on $U$.
\Cref{lem:moment-estimated-total-degree} proves that when $U$ approximates $W$ within a constant factor, $\E[Z_q(U)] = \Theta(M_s)$ and the variance bound is preserved up to a constant factor.
This allows us to use a coarse estimate $\widetilde{W}$ when searching for $q$ and an accurate estimate of $W$ only in the final estimator (\Cref{alg:moment-estimation}).

\subparagraph*{Choosing the expected sample size.}
We first invoke \Cref{alg:endpoint-capture} with accuracy $1/4$ and let $\widetilde{W}$ be twice its output, so that $\widetilde{W} \in [3W/4, 5W/4]$ with probability at least $63/64$.
Because $M_s \le W^s$, the quantity $(2\widetilde{W})^s$ bounds $M_s$ from above whenever $\widetilde{W}$ is accurate.
Starting from this upper bound, we search for $M_s$ by testing geometrically decreasing candidate values $\Gamma_j := (2\widetilde{W})^s 2^{-j}$.
Each candidate $\Gamma_j$ corresponds to a candidate $A_j := \widetilde{W} / \Gamma_j^{1/s} = 2^{j/s - 1}$ for $\Lambda_s$.
At iteration $j$, we evaluate the test statistic $Z_j := Z_{4s A_j}(\widetilde{W})$ using $4s A_j = O(A_j)$ expected samples, terminating at the first index $J$ where $Z_J \ge \Gamma_J / 2$.

When $\widetilde{W}$ is accurate, each test statistic has expectation $\E[Z_j \mid \widetilde{W}] = O(M_s)$.
Markov's inequality ensures that accepting an oversized candidate $\Gamma_j \gg M_s$ is unlikely.
Conversely, once $\Gamma_j$ falls below a sufficiently small constant multiple of $M_s$, the expected sample size $4s A_j$ makes the variance sufficiently small that Chebyshev's inequality guarantees acceptance with high constant probability.
\Cref{lem:moment-search} shows that when $\widetilde{W}$ is accurate, the selected values satisfy $\Gamma_J = O(M_s)$, hence $A_J = \Omega(\Lambda_s)$, with high constant probability.
Substituting $A_J$ for $\Lambda_s$ determines the final expected sample size:
\[
  q := D_s\left(
    1 + \min\left\{
      \frac{A_J}{\eps^{2/s}},
      \frac{A_J^{2s/(2s-1)}}{\eps^{2/(2s-1)}}
    \right\}
  \right),
\]
where $D_s \ge 1$ is a constant depending only on $s$.
\Cref{lem:moment-search} additionally establishes that the search terminates almost surely for every $\widetilde{W} > 0$ and that its expected cost is $O(\Lambda_s)$.

\begin{algorithm}[t]
\caption{Estimating $M_s$}
\label{alg:moment-estimation}
\begin{algorithmic}[1]
\Require Endpoint access to $G$; accuracy $0<\eps\le1$; integer $s\ge2$.
\Statex Fix a sufficiently small constant $\kappa>0$ and a sufficiently large constant $D_s\ge1$ depending only on $s$.
\State Run \Cref{alg:endpoint-capture} with accuracy $1/4$; let $\widetilde W$ be twice its output.
\For{$j=0,1,2,\ldots$}
  \State $\Gamma_j\gets(2\widetilde W)^s2^{-j}$, $A_j\gets\widetilde W/\Gamma_j^{1/s}=2^{j/s-1}$.
  \State Compute $Z_j \gets Z_{4sA_j}(\widetilde{W})$ by~\eqref{eq:moment-estimated-total-degree}.
  \If{$Z_j\ge \Gamma_j/2$}
    \State $J\gets j$ and \textbf{break};
  \EndIf
\EndFor
\State $\displaystyle q\gets D_s\!\left(1+\min\!\left\{\frac{A_J}{\eps^{2/s}},\frac{A_J^{2s/(2s-1)}}{\eps^{2/(2s-1)}}\right\}\right)$.
\State $\eta(q)\gets\kappa\max\{\eps,q^{-1/2}\}$.
\State Run \Cref{alg:endpoint-capture} with accuracy $\eta(q)$; let $\widehat W$ be twice its output.
\State Compute and \Return $Z_q(\widehat{W})$ by~\eqref{eq:moment-estimated-total-degree}.
\end{algorithmic}
\end{algorithm}

\subparagraph*{Choosing the accuracy for $W$.}
With the expected sample size $q$ fixed, it remains to choose the accuracy of the estimate $\widehat{W}$ of $W$ used in the final estimator.
While relative error $O(\eps)$ for $W$ gives sufficient accuracy, invoking \Cref{alg:endpoint-capture} with accuracy $\Theta(\eps)$ costs $\Omega(\eps^{-2})$ queries, which dominates $q$ when $q = o(\eps^{-2})$.
Such high accuracy is unnecessary because high degree vertices ($q d_v / W \gg 1$) appear with probability near $1$, and when $U$ approximates $W$ within a constant factor, their contributions $(d_v)^s / \pi_v(q, U) \approx (d_v)^s$ are largely insensitive to $U$.
\Cref{lem:moment-estimated-total-degree} quantifies this insensitivity: any estimate $U$ with $|U/W - 1| \le \eta \le 1/2$ incurs relative bias $\abs{\E[Z_q(U)] - M_s} / M_s = O(\eta \min\{1, \sqrt{q R_s(q/2)}\})$, which reduces to $O(\eta \min\{1, \eps \sqrt{q}\})$ once $R_s(q/2) \le \eps^2$ after a successful search for $q$.
Therefore setting $\eta(q) := \kappa \max\{\eps, q^{-1/2}\}$ for a small constant $\kappa > 0$ guarantees relative bias at most $\eps / 2$.
When $q = o(\eps^{-2})$, this relaxes the required accuracy to $\eta(q) \asymp q^{-1/2} \gg \eps$.
We invoke \Cref{alg:endpoint-capture} with accuracy $\eta(q)$, set $\widehat{W}$ to twice its output, and return $Z_q(\widehat{W})$.

For every $0 < \eps \le 1$, \Cref{alg:moment-estimation} is advice-free, terminates almost surely, returns a $(1\pm\eps)$-estimate of $M_s$ with probability at least $15/16$, and has the expected query complexity stated in \Cref{thm:intro-moment}.
The complete analysis appears in \Cref{sec:moment-detail}.

\section{Lower Bounds for Degree Moment Estimation}
\label{sec:moment-lower}
We prove \Cref{thm:intro-moment-lb} using two lower bound constructions for two regimes of $k\eps^2$ separated by a threshold $\kappa_s > 0$ depending only on $s$.
Fix an integer $s \ge 2$, a bound $k \ge 2$ on the matching number, and an accuracy parameter $0 < \eps \le \eps_0(s)$.
For each regime, we construct distributions $\mathsf{H}_0$ and $\mathsf{H}_1$ over graphs with $\nu(G) \le k$, and apply the lower bound strategy of \Cref{sec:edge-lower} with parameter $\theta(G) = M_s(G)$:
\begin{itemize}
  \item When $k\eps^2 \ge \kappa_s$, the \emph{stars and isolated edges} construction (\Cref{sec:moment-stars}) establishes an expected query lower bound of $\Omega(k^{1-1/s}/\eps^{2/s})$.
  \item When $k\eps^2 \le \kappa_s$, the \emph{paired centers} construction (\Cref{sec:moment-paired}) establishes an expected query lower bound of $\Omega(k^{1-1/(2s-1)}/\eps^{2/(2s-1)})$.
\end{itemize}
Within its regime, each of these bounds attains the minimum in \Cref{thm:intro-moment-lb} up to a constant factor depending only on $s$.
Thus choosing $\eps_0(s)$ sufficiently small for both constructions proves the theorem for all $k\ge2$.

\subsection{Stars and Isolated Edges}
\label{sec:moment-stars}
We first establish the lower bound $\Omega(k^{1-1/s}/\eps^{2/s})$ in the regime $k\eps^2 \ge \kappa_s$.

\subparagraph*{Construction of stars and isolated edges.}
Consider a graph consisting of $D$-stars ($K_{1,D}$) and isolated edges.
A $D$-star contributes $D^s + D$ to the moment $M_s$, whereas $D$ isolated edges contribute only $2D$.
Thus replacing $D$ isolated edges with a single $D$-star increases $M_s$ without altering the total edge count $m$.
Balancing the contributions of stars and isolated edges, we compare $G_0$ consisting of $a$ stars and $L$ isolated edges against $G_1$ consisting of $a + r$ stars and $L - rD$ isolated edges, where $L := aD^s$.
Specifically, we set
\[
  a := \lceil\eps^{-2}\rceil, \qquad
  r := \lceil 32\eps a \rceil, \qquad
  D := \left\lfloor\left(k / (4a)\right)^{1/s}\right\rfloor, \qquad
  L := aD^s.
\]
For sufficiently small $\eps_0(s)$ and sufficiently large $\kappa_s$, these parameters satisfy $D \ge 2$, $r \le a/2$, and $L - rD \ge L/2$.

We call $G_0$ and $G_1$ the base graphs and pad them with isolated vertices to share a common vertex set. 
The distributions $\mathsf{H}_0$ and $\mathsf{H}_1$ are induced by uniform random vertex relabelings of $G_0$ and $G_1$.
Both base graphs have matching number $\nu(G_i) \le a + L \le k$ and the same total edge count $m = aD + L = a(D^s + D)$.
Since $M_s(G_1) - M_s(G_0) = r(D^s - D) \ge 4\eps M_s(G_0)$, we have $(1+\eps)M_s(G_0) < (1-\eps)M_s(G_1)$, so the moments $M_s(G_0)$ and $M_s(G_1)$ are $(1\pm\eps)$-separated (\Cref{lem:stars-and-edges-parameters}).

\subparagraph*{Indistinguishability of stars and isolated edges.}
To prove indistinguishability, we consider an augmented oracle that returns each sampled edge together with the labeled component containing it.
Since this component determines the degrees of all its vertices, $q$ responses from this augmented oracle suffice to simulate any endpoint access algorithm making at most $q$ queries via randomized postprocessing.
Because both base graphs are uniformly relabeled and share the same edge count $m$, \Cref{lem:moment-stars-component-counts} shows that when the number of queries is drawn from a Poisson distribution, the responses carry no distinguishing information beyond the numbers of stars and isolated edges that are \emph{observed} (hit by at least one query).
Let $\mathsf P_i$ denote the joint distribution of the responses to $q$ queries to this oracle under $\mathsf H_i$.

Since $q$ queries observe each star with probability $p_q := 1 - (1 - D/m)^q$, the expected number of observed stars is $a p_q$ under $\mathsf{H}_0$ and $(a + r) p_q$ under $\mathsf{H}_1$.
For $q \le c_s a D^{s-1}$ with a sufficiently small constant $c_s > 0$, we have $qD/m \le c_s$, so $p_q = \Theta(q / (a D^{s-1}))$.
The number of observed stars has standard deviation $\Theta(\sqrt{a p_q}) = \Theta(\sqrt{q / D^{s-1}})$ under $\mathsf{H}_0$, while its expectations under the two hypotheses differ by $r p_q = \Theta(\eps a p_q) = \Theta(\eps q / D^{s-1})$.
Thus the ratio of this difference to the standard deviation is $\Theta(\eps \sqrt{a p_q}) = O(\sqrt{c_s})$. 
Formalizing this comparison, \Cref{lem:moment-stars-fixed-sample-tv} proves that $\TV(\mathsf{P}_0, \mathsf{P}_1) \le 1/4$ for all $1 \le q \le c_s a D^{s-1}$, which yields the worst case expected query lower bound $\Omega(aD^{s-1}) = \Omega(k^{1-1/s}/\eps^{2/s})$ as $a = \Theta(\eps^{-2})$ and $D = \Theta((k\eps^2)^{1/s})$. 
The complete proof appears in \Cref{sec:moment-stars-detail}.

\subsection{Paired Centers}
\label{sec:moment-paired}
To complete the proof of \Cref{thm:intro-moment-lb}, we establish the complementary lower bound $\Omega(k^{1-1/(2s-1)}/\eps^{2/(2s-1)})$ in the regime $k\eps^2 \le \kappa_s$.

\subparagraph*{Construction of paired centers.}
Let $k' := 2\lfloor k/2\rfloor \le k$.
We consider $k'$ \emph{center} vertices partitioned into $k'/2$ disjoint pairs.
To each center pair $(u,v)$, we attach $D \ge 2$ local \emph{gadgets}, giving $N := k'D/2$ gadgets in total.
Each gadget introduces two auxiliary vertices $x,y$ and selects between two configurations:
\begin{description}
  \item[Shared] Insert edges $\{u,x\}$ and $\{v,x\}$, so $x$ has degree $2$ and $y$ is isolated.
  \item[Separate] Insert edges $\{u,x\}$ and $\{v,y\}$, so both auxiliary vertices have degree $1$.
\end{description}
In both configurations, one edge is incident to $u$ and one to $v$.
Thus every center has degree $D$, and every graph has $m = k'D = 2N$ edges.
Switching a gadget from separate to shared increases the contribution of its auxiliary vertices to $M_s$ from $2$ to $2^s$.
The centers form a vertex cover of size $k'$, and $D \ge 2$ lets us match both centers in each pair using distinct gadgets, so $\nu(G) = k' \le k$.
For $0 < \beta \le 1/4$, we define distributions $\mathsf H_0$ and $\mathsf H_1$ over graphs on the same fixed set of $k' + 2N$ labeled vertices: independently for each of the $N$ gadgets, $\mathsf H_0$ chooses \emph{shared} with probability $1/2-\beta$ and \emph{separate} otherwise, while $\mathsf H_1$ chooses \emph{shared} with probability $1/2+\beta$ and \emph{separate} otherwise.

Let $Z$ denote the number of \emph{shared} gadgets. 
The moment evaluates to $M_s = k'D^s + 2(N-Z) + 2^s Z = k'D^s + 2N + (2^s-2)Z = \Theta(k'D^s)$. 
\Cref{lem:paired-centers-moments} proves that the expected moments under the two distributions differ by $2(2^s-2)\beta N$ while their standard deviations are $O(\sqrt{N})$.
We choose $D$ so that the additive error $\eps M_s = \Theta(\eps k'D^s)$ allowed by a $(1\pm\eps)$-estimate is comparable to $\sqrt{N}$, and choose $\beta$ so that the difference of the expected moments is a sufficiently large multiple of $\sqrt{N}$.
Since $N = k'D/2$, balancing $\eps k' D^s \asymp \sqrt{k'D}$ gives $D = \Theta((k'\eps^2)^{-1/(2s-1)})$.
Specifically, we set
\[
  D := \max\!\left\{2,\left\lceil(k'\eps^2)^{-1/(2s-1)}\right\rceil\right\},
  \qquad
  \beta := \frac{\beta_s}{\sqrt{N}},
\]
where $\beta_s$ is a sufficiently large constant.
For sufficiently small $\eps_0(s)$, \Cref{lem:moment-paired-separation} verifies $\beta \le 1/4$ and gives $(1\pm\eps)$-separated sets $I_0, I_1$ of moment values such that $\Prob_{\mathsf H_i}[M_s \in I_i] \ge 99/100$ for each $i \in \{0,1\}$.

\subparagraph*{Indistinguishability of paired centers.}
Since every gadget contributes two edges under either configuration, a uniform random edge can be sampled by first selecting one of the $N$ gadgets uniformly at random, and then choosing one of its two edges uniformly.
To prove indistinguishability, we consider an augmented oracle that reveals the selected gadget and its configuration (shared or separate).
Because every center has degree $D$ and the degrees of returned auxiliary vertices are determined by the revealed configuration ($2$ for shared, $1$ for separate), $q$ responses from this augmented oracle suffice to simulate any endpoint access algorithm making at most $q$ queries via randomized postprocessing.

Let $\mathsf P_i$ denote the joint distribution of the responses to $q$ queries to this oracle under $\mathsf H_i$.
Conditional on the history, a newly selected gadget is shared with probability $1/2-\beta$ under $\mathsf H_0$ and $1/2+\beta$ under $\mathsf H_1$, while repeat selections of a gadget reveal no new information.
Thus each distinct selected gadget contributes $\KL(\operatorname{Ber}(1/2-\beta) \| \operatorname{Ber}(1/2+\beta)) = O(\beta^2)$, and the total divergence across $q$ queries is $O(q\beta^2)$.
Pinsker's inequality then gives $\TV(\mathsf P_0, \mathsf P_1) = O(\beta\sqrt{q}) = O(\sqrt{q/N})$ (\Cref{lem:moment-paired-query-tv}).
Setting $q := \lfloor c_s N \rfloor$ for a sufficiently small constant $c_s > 0$ ensures $\TV(\mathsf P_0, \mathsf P_1) \le 1/4$, yielding the worst case expected query lower bound $\Omega(N) = \Omega(k^{1-1/(2s-1)}/\eps^{2/(2s-1)})$.
The complete proof appears in \Cref{sec:moment-paired-detail}.

\section{Conclusion}
In this paper, we presented advice-free algorithms for estimating the edge count $m$ and the degree moments $M_s$ under endpoint access, showing that the query complexity of both problems is governed by the matching number $\nu$ rather than the vertex count $n$.
For degree moments, our algorithm also improves the dependence on $\eps$ polynomially over prior work.
Matching these upper bounds, our lower bounds prove that the resulting query complexities are optimal up to constant factors in the worst case over graphs of bounded matching number.

\subparagraph*{Future work.}
Our results suggest two compelling directions for future work.
While the matching number characterizes the worst case over $\mathcal{G}_{\le k}$, our algorithms can run substantially faster on individual graphs with favorable degree concentration (\Cref{sec:moment-instance-sensitive}).
A natural direction is thus to identify structural parameters beyond the matching number that capture this behavior, and to design advice-free algorithms that are optimal in terms of them.
Moreover, while our optimality guarantees establish the worst case query complexity at constant success probability, characterizing the exact complexity tradeoff in terms of failure probability $\delta$ also remains an intriguing open problem.

\subparagraph*{Disclosure of AI usage.}
The authors disclose that generative AI tools (specifically OpenAI GPT, Anthropic Claude and Google Gemini) were used during the preparation of this manuscript.
While all research directions, algorithm designs, structural insights, and proof architectures were conceived and directed by the authors, AI tools assisted with text polishing, intermediate algebraic derivations, evaluating explicit constants, and checking proofs and statements about prior work.
The authors assume full and sole responsibility for all content and results in this manuscript.

\newpage
\bibliography{references}

\newpage
\appendix

\section{Analysis of Algorithm~\ref{alg:endpoint-capture}}
\label{sec:endpoint-capture-detail}
In this section, we provide the omitted analysis of \Cref{alg:endpoint-capture}, completing the proof of \Cref{thm:endpoint-capture}.
Throughout this section, fix an accuracy parameter $0 < \eps \le 1$, set $R := \lceil 64\eps^{-2} \rceil$, and write $\nu := \nu(G)$, $W := 2m$, $D(S) := \sum_{v\in S} d_v$, and $p_S := D(S)/W$.
At a reached iteration with sample size $q$, the first $q$ fair endpoint samples define $S$, and the second $q$ samples give the hit count $H$.
The algorithm queries the degrees of the vertices in $S$ to compute $D(S)$ only when $H \ge R$, then draws more fair endpoint samples until observing $R$ hits in $S$, where $T$ denotes the number of trials required to observe $R$ hits.
The history records all oracle responses and private random choices.
All probabilities and expectations in the lemmas below are conditional on an arbitrary history up to the start of the iteration under consideration, unless stated otherwise.

\subparagraph*{Accuracy and ratio moments.}
\begin{lemma}
\label{lem:capture-return-law}
Fix any history up to the degree queries at the first iteration where $H\ge R$.
Let $T$ be the number of trials required to observe $R$ hits in $S$, and let $\widehat m := D(S)T/(2R)$.
Then $T<\infty$ almost surely, and $Z:=\widehat m/m=p_ST/R$ satisfies $\E[Z]=1$ and
\[
  \Prob\!\left[\abs{Z - 1} \ge \eps\right] \le \frac{1}{64}, \qquad
  \E[Z^2]=O(1), \qquad
  \E[Z^{-2}]=O(1).
\]
\end{lemma}

\begin{proof}
Because $q\ge1$ and every sampled endpoint has positive degree, we have $p_S=D(S)/(2m)>0$.
As the trials are independent fair endpoint samples, $T$ follows a negative binomial distribution:
\[
  \Prob[T = t] = \binom{t-1}{R-1} p_S^R (1 - p_S)^{t-R}, \qquad t = R, R+1, R+2, \ldots.
\]
Because $p_S>0$, $T$ is finite almost surely.
The mean and variance of $T$ are
\[
  \E[T] = \frac{R}{p_S}, \qquad \Var[T] = \frac{R(1 - p_S)}{p_S^2}.
\]
Since $Z=p_ST/R$, we have $\E[Z]=1$.
Applying Chebyshev's inequality with $R=\lceil64\eps^{-2}\rceil$, we obtain
\[
  \Prob\!\left[\abs{Z - 1} \ge \eps\right] \le \frac{\Var[Z]}{\eps^2} = \frac{1 - p_S}{R \eps^2} \le \frac{1}{64}.
\]
Moreover, the second moment satisfies
\[
  \E[Z^2]=\Var[Z]+(\E[Z])^2=\frac{1-p_S}{R} + 1=O(1).
\]
For the inverse second moment, the negative binomial mass function gives
\[
  \E\!\left[\frac{1}{(T-1)(T-2)}\right]
  =\frac{p_S^2}{(R-1)(R-2)}
    \sum_{t=R}^{\infty}\binom{t-3}{R-3}p_S^{R-2}(1-p_S)^{t-R}
  =\frac{p_S^2}{(R-1)(R-2)}.
\]
The summands are the probabilities that $R-2$ hits require $t-2$ trials, so their sum is one.
Since $T\ge R\ge64>2$, we have $T^{-2}\le((T-1)(T-2))^{-1}$, and hence
\[
  \E[Z^{-2}]=\frac{R^2}{p_S^2}\E[T^{-2}]
  \le\frac{R^2}{(R-1)(R-2)}=O(1). \qedhere
\]
\end{proof}

\subparagraph*{Continuation probability.}
We bound the mean and variance of $H$, then apply Chebyshev's inequality to control the probability of continuing to the next iteration.

\begin{lemma}
\label{lem:captured-endpoint-mass}
Consider any reached iteration with sample size $q$.
Let $S$ be the set of distinct vertices in the first $q$ fair endpoint samples, and let $H$ count hits in $S$ among the second batch of $q$ samples.
Then
\[
  \E[p_S]\ge\frac{q}{4(2\nu+q)},
  \qquad
  \E[H]\ge\frac{q^2}{4(2\nu+q)},
  \qquad \Var[H]\le2\E[H].
\]
\end{lemma}

\begin{proof}
Let $C^\star$ be the set of endpoints of a maximum matching.
Since $C^\star$ is a vertex cover of size $2\nu$ (as an uncovered edge would extend the matching), its endpoint mass satisfies $p_{C^\star} \ge 1/2$.
For each $v\in V$, let $I_v:=\Ind[v\in S]$, so that $p_S=\sum_{v\in V}p_vI_v$ and $\E[I_v]=1-(1-p_v)^q$.
Write $\mu:=\E[p_S]$.
Restricting the sum to $C^\star$, using $1-(1-x)^q\ge1-e^{-qx}\ge qx/(1+qx)$, and applying Cauchy--Schwarz bounds the expected endpoint mass as
\[
  \mu \ge \sum_{v \in C^\star} p_v \bigl(1 - (1-p_v)^q\bigr)
  \ge q \sum_{v \in C^\star} \frac{p_v^2}{1 + q p_v}
  \ge \frac{q p_{C^\star}^2}{|C^\star| + q p_{C^\star}}
  \ge \frac{q}{4(2\nu + q)}.
\]

To bound the variance of $p_S$, we use the nonpositive covariance of $I_u$ and $I_v$ for distinct vertices $u,v\in V$:
\[
  \Cov[I_u, I_v] = (1 - p_u - p_v)^q - (1 - p_u)^q (1 - p_v)^q \le 0,
\]
which follows from $1-p_u-p_v\le(1-p_u)(1-p_v)$.
Since $q p_v(1-p_v)^q\le p_v\sum_{i=0}^{q-1}(1-p_v)^i=1-(1-p_v)^q$, we obtain
\[
  \Var[p_S] \le \sum_{v \in V} p_v^2 \Var[I_v]
  \le \sum_{v \in V} p_v^2 (1 - p_v)^q
  \le \frac{1}{q} \sum_{v \in V} p_v \bigl(1 - (1 - p_v)^q\bigr) = \frac{\mu}{q}.
\]
The second batch is independent of the first, so conditional on $S$, we have $H\sim\operatorname{Bin}(q,p_S)$.
Consequently, $\E[H]=q\mu\ge q^2/(4(2\nu+q))$, and the law of total variance gives
\[
  \Var[H]
  = \E[\Var[H\mid S]]+\Var[\E[H\mid S]]
  \le q\mu+q^2\Var[p_S]
  \le 2q\mu=2\E[H]. \qedhere
\]
\end{proof}

\begin{corollary}
\label{lem:capture-hit-threshold}
At any reached iteration with sample size $q\ge Q:=64(R+\sqrt{\nu R})$, $\Prob[H<R]\le1/4$.
\end{corollary}

\begin{proof}
Since the assumption $q \ge Q$ implies $q\ge64R$ and $q\ge64\sqrt{\nu R}$, we have
\[
  32qR+64\nu R\le\frac{q^2}{2}+\frac{q^2}{64}<q^2.
\]
The mean bound in \Cref{lem:captured-endpoint-mass} therefore implies
\[
  \E[H]\ge\frac{q^2}{4(2\nu+q)}\ge8R.
\]
Combining Chebyshev's inequality with the variance bound in \Cref{lem:captured-endpoint-mass} gives
\[
  \Prob[H<R]
  \le\frac{\Var[H]}{(\E[H]-R)^2}
  \le\frac{2\E[H]}{(\E[H]-R)^2}
  \le\frac{8}{\E[H]}
  \le\frac1R<\frac14,
\]
where we used $\E[H]-R\ge\E[H]/2$ and $R\ge64$.
\end{proof}

\subparagraph*{Expected query complexity.}
Obtaining $R$ hits in $S$ can be expensive when $p_S$ is small, but the probability of observing $H\ge R$ is correspondingly small.

\begin{lemma}
\label{lem:capture-fixed-scale-cost}
At any reached iteration with sample size $q$, let $T$ count the trials required to observe $R$ hits in $S$, with $T:=0$ if $H<R$.
Then $\E[T\mid S]\le q$ for every possible $S$, and the expected query cost of the entire iteration is at most $4q$.
\end{lemma}

\begin{proof}
Every possible $S$ has $p_S>0$ as $q \ge 1$.
If $H<R$, no further samples are drawn; otherwise, the algorithm uses independent fair endpoint samples to observe $R$ hits in $S$. 
Thus $\E[T\mid S,H]=\Ind[H\ge R] \cdot R / p_S$. 
Applying $\E[H\mid S]=qp_S$ and Markov's inequality yields
\[
  \E[T\mid S]=\Prob[H\ge R\mid S]\frac{R}{p_S}
  \le\frac{qp_S}{R}\cdot\frac{R}{p_S}=q.
\]
Taking the expectation over $S$ gives $\E[T]\le q$.
The two initial batches cost $2q$ random edge queries, computing $D(S)$ when $H\ge R$ costs $\abs{S}\le q$ degree queries, and the final sampling step uses at most $q$ queries in expectation.
The expected cost of the entire iteration is therefore at most $4q$.
\end{proof}
We now combine the continuation and cost bounds to complete the proof of \Cref{thm:endpoint-capture}.

\begin{proof}[Proof of \Cref{thm:endpoint-capture}]
Let $q_\star$ be the smallest power of two satisfying $q_\star \ge Q = 64(R + \sqrt{\nu R})$, so that $q_\star < 2Q$.
By \Cref{lem:capture-fixed-scale-cost}, the expected contribution of each iteration with sample size $q$ is at most $4q$ times the probability of reaching it.
In particular, the iterations before $q_\star$ contribute at most $\sum_{j:\,2^j<q_\star}4\cdot2^j<4q_\star$ queries in expectation.

For every integer $\ell \ge 0$, reaching $2^{\ell+1}q_\star$ requires reaching $2^\ell q_\star$ and observing $H<R$ there.
By \Cref{lem:capture-hit-threshold}, the latter event has conditional probability at most $1/4$ for every history reaching that iteration.
Averaging over these histories gives
\[
  \Prob[\text{reaches }2^{\ell+1}q_\star]
  \le \frac14\Prob[\text{reaches }2^\ell q_\star].
\]
Since $\Prob[\text{reaches }q_\star]\le1$, induction yields
\begin{equation}
  \Prob[\text{reaches }2^\ell q_\star]\le4^{-\ell}, \qquad \ell\ge0.
  \label{eq:capture-reach}
\end{equation}
Thus the total expected query cost satisfies
\[
  \E[\text{total query cost}]
  \le 4q_\star + \sum_{\ell=0}^{\infty} 4^{-\ell} \cdot 4\left(2^\ell q_\star\right)
  = 4q_\star + 4q_\star \sum_{\ell=0}^{\infty} 2^{-\ell}
  = O(q_\star) = O(Q).
\]
Substituting $R=\lceil64\eps^{-2}\rceil$ into $Q=64(R+\sqrt{\nu R})$, we obtain a total expected query complexity of $O\left(R + \sqrt{\nu R}\right) = O\!\left(1/\eps^2 + \sqrt{\nu}/\eps\right)$.
By \eqref{eq:capture-reach}, the probability of reaching $2^\ell q_\star$ tends to zero as $\ell\to\infty$, so an iteration with $H\ge R$ occurs eventually almost surely.
Conditioning on the history up to the degree queries at this terminating iteration, \Cref{lem:capture-return-law} guarantees that the final sampling step terminates almost surely and satisfies the accuracy and moment bounds.
Averaging over all such histories yields $\E[\widehat m]=m$, the $(1\pm\eps)$-estimate guarantee with probability at least $63/64$, and the unconditional moment bounds $\E[(\widehat m/m)^2]=O(1)$ and $\E[(m/\widehat m)^2]=O(1)$.
\end{proof}

\section{Analysis of Algorithm~\ref{alg:sampled-cover}}
\label{sec:sampled-cover-detail}
In this section, we provide the omitted analysis of \Cref{alg:sampled-cover}, completing the proof of \Cref{thm:sampled-cover}.
We prove the theorem by bounding both the uncovered edge count $r(C)$ and the error in estimating the internal edge count $h(C)$ by $\eps m/2$.
Throughout this section, write $m:=|E(G)|>0$ and $\nu:=\nu(G)$.

\subparagraph*{Constructing $C$ by greedy matching.}
We first establish that the greedy matching loop is unlikely to terminate at any scale $b$ while more than $\eps m / 2$ edges remain uncovered by $C = V(M_b)$.

\begin{lemma}
\label{lem:sampled-approx-cover}
For every $0 < \eps \le 1$, the greedy matching loop in \Cref{alg:sampled-cover} terminates at a scale $b \le 2\nu$ on every execution using $O(\nu/\eps)$ random edge queries. 
The vertex set $C$ satisfies $|C| < 2b$ and
\[
  \Prob\left[r(C) > \frac{\eps m}{2}\right] < \frac{1}{12}.
\]
\end{lemma}

\begin{proof}
Fix a scale $b$ reached by the algorithm and condition on the preceding history.
Within the iteration at scale $b$, let $M_b(i)$ denote the matching after $i$ samples ($0 \le i \le q_b$), with $M_b(0) = \varnothing$ and the completed matching $M_b := M_b(q_b)$, and write $r_i := r(V(M_b(i)))$.
At step $i + 1$, the sampled edge enlarges the current matching if and only if it is one of the $r_i$ uncovered edges.
The sequence $r_i$ is nonincreasing because matched vertices remain matched throughout the iteration.

To compare the number of insertions with a binomial random variable, we consider the following equivalent sampling process using independent $U_1,\dots,U_{q_b}\sim\operatorname{Unif}(0,1)$.
At step $i+1$, choose an uncovered edge if $U_{i+1}\le r_i/m$ and a covered edge otherwise, uniformly within the chosen class.
This generates a uniform random edge independent of the earlier steps, and the edge is inserted when it is uncovered.
On the event $r(V(M_b)) > \eps m / 2$, every step has $r_i/m > \eps/2$, so every occurrence of $U_{i+1}\le\eps/2$ causes an insertion.
Consequently,
\[
  \Prob\left[|M_b|<b,\ r(V(M_b))>\frac{\eps m}{2}\right]
  \le \Prob\left[\operatorname{Bin}(q_b,\eps/2)<b\right]
  \le e^{-3b},
\]
where the last inequality follows from the multiplicative Chernoff bound, using $q_b\eps/2\ge8b$.
Averaging over the histories reaching each scale and taking a union bound over the doubling scales gives
\[
  \Prob\left[r(C)>\frac{\eps m}{2}\right]
  \le \sum_{j\ge0}e^{-3\cdot2^j}<\frac1{12}.
\]

Finally, every matching has size at most $\nu$, so the loop stops no later than the first power of two $b_\star>\nu$.
Since $m>0$ implies $\nu\ge1$, we have $b_\star\le2\nu$.
Therefore the total number of random edge queries on every execution is at most
\[
  \sum_{2^j\le b_\star}\left\lceil\frac{16\cdot2^j}{\eps}\right\rceil
  =O(b_\star/\eps)=O(\nu/\eps),
\]
where $\eps\le1$ absorbs the rounding costs. Since the loop stops only when $|M_b|<b$, the set $C$ satisfies $|C|=2|M_b|<2b$.
\end{proof}

\subparagraph*{Estimating the internal edge count.} 
It remains to bound the estimation error for $h(C)$. 
\begin{lemma}
\label{lem:internal-correction}
Condition on any history fixing $b \ge 1$ and $C \subseteq V$ with $|C| < 2b$.
There is a constant $K$ such that, for every $0 < \eps \le 1$, \Cref{alg:internal-edges} uses $t = \lceil Kb/\eps \rceil$ random edge queries and returns an estimate $\widehat h$ satisfying
\[
  \Prob\left[|\widehat{h} - h(C)| > \frac{\eps m}{2}\right] < \frac{1}{12}.
\]
\end{lemma}
To prove \Cref{lem:internal-correction}, we first establish that as long as the number of internal samples reaches $24b/\eta + 1$, the estimator $1 / Y$ of \Cref{alg:internal-edges} achieves relative error $\eta$ with probability at least $23/24$.
This estimator and its analysis are along the lines of Beretta and T\v{e}tek~\cite{BerettaTetek2024}, who count collisions among weighted elements.
Here the samples are internal edges, each equally likely, and the bound $h < 2b^2$ replaces their advice on the number of elements.

\begin{lemma}
\label{lem:collision-reciprocal}
Let $h \ge 1$ and $b \ge 1$ satisfy $h < 2b^2$, and let $0 < \eta \le 1$.
Suppose we draw $N$ independent uniform samples $e_1, \dots, e_N$ with replacement from a set of $h$ edges, where $N - 1 \ge 24b/\eta$.
Let $Y := \binom{N}{2}^{-1} \sum_{1 \le i < j \le N} \Ind[e_i = e_j]$ denote the fraction of colliding pairs, and set $\widehat h := 1/Y$ if $Y > 0$ and $\widehat h := 0$ otherwise.
Then
\[
  \Prob[|\widehat h-h|>\eta h]<\frac1{24}.
\]
\end{lemma}

\begin{proof}
Write $I_{ij}:=\Ind[e_i=e_j]$, so $\E[I_{ij}]=1/h$.
Indicators indexed by disjoint pairs are independent, while for distinct $i,j,k$,
\[
  \E[I_{ij} I_{ik}] = \Prob[e_i = e_j = e_k] = \frac{1}{h^2} = \E[I_{ij}] \E[I_{ik}].
\]
Hence, all distinct collision indicators are pairwise uncorrelated, which implies
\begin{equation}
  \E[Y] = \frac{1}{h}, \qquad \Var[Y] = \frac{1}{\binom{N}{2}^2} \sum_{1 \le i < j \le N} \Var[I_{ij}] \le \frac{1}{h \binom{N}{2}}.
  \label{eq:collision-variance}
\end{equation}
Chebyshev's inequality and \eqref{eq:collision-variance} give
\begin{equation}
  \Prob\left[\left|Y - \frac{1}{h}\right| > \frac{\eta}{2h}\right] \le \frac{4h}{\eta^2\binom{N}{2}} < \frac{1}{24},
  \label{eq:collision-concentration}
\end{equation}
where the last inequality uses $\binom{N}{2} \ge (N-1)^2/2$, $N-1 \ge 24b/\eta$, and $h < 2b^2$.
On the concentration event $|Y-1/h|\le\eta/(2h)$, write $Y=(1+\theta)/h$ with $|\theta|\le\eta/2\le1/2$.
Then $Y>0$, and its reciprocal satisfies
\[
  |\widehat{h} - h| = \left|\frac{1}{Y} - h\right| = \frac{h|\theta|}{|1 + \theta|} \le \frac{h(\eta/2)}{1 - 1/2} = \eta h.
\]
Applying \eqref{eq:collision-concentration} completes the proof.
\end{proof}

\begin{proof}[Proof of \Cref{lem:internal-correction}]
Write $h := h(C)$ and $p := h/m$.
The algorithm draws $t = \lceil Kb/\eps \rceil$ random edges, so the number $N$ of internal samples follows $\operatorname{Bin}(t, p)$.
If $h = 0$, then $p = 0$, which implies $N = 0 < t\eps/4$; the algorithm returns $\widehat h = 0 = h$.

Assume henceforth that $h \ge 1$, so $1 \le h \le \binom{|C|}{2} < 2b^2$.
Conditional on $N$, the internal samples are independent and uniformly distributed over the $h$ internal edges of $G[C]$, allowing us to apply \Cref{lem:collision-reciprocal}.
Whenever $N \ge t\eps/4$, we have $N \ge Kb/4$, so choosing $K$ sufficiently large ensures $N - 1 \ge 24b \ge 24$; thus when the algorithm computes $Y$, we have $N \geq 2$ and hence $Y$ is well defined. 

Suppose first that $p \le \eps / 2$, so that $h \le \eps m / 2$.
If $N < t\eps/4$, the algorithm returns $\widehat h = 0$, incurring error $|\widehat h - h| = h \le \eps m / 2$.
If $N \ge t\eps/4$ instead, \Cref{lem:collision-reciprocal} with $\eta = 1$ ensures that $|\widehat h - h| \le h \le \eps m / 2$, except with probability less than $1/24$.

Now suppose that $p > \eps / 2$.
Because $b \ge 1$ and $p/\eps > 1/2$, the expected number of internal samples satisfies $\E[N] = tp \ge Kbp/\eps > K/2$.
The multiplicative Chernoff bound gives $\Prob\left[N < tp/2\right] \le e^{-tp/8} < e^{-K/16} < 1/24$ for sufficiently large $K$.
On the complementary event $N \ge tp/2 > t\eps/4$, the algorithm evaluates $Y$.
To achieve additive error $\eps m / 2$, set $\eta := \eps m / (2h) = \eps / (2p) < 1$.
Because $bp/\eps > 1/2$, choosing $K$ sufficiently large ensures that on this event,
$N - 1 \ge Kbp /(2\eps) - 1 \ge 48bp/\eps = 24b/\eta$.
For every such $N$, \Cref{lem:collision-reciprocal} guarantees $|\widehat h-h|\le\eta h=\eps m/2$ with failure probability less than $1/24$.
By a union bound over these two events, the total failure probability is less than $1/24 + 1/24 = 1/12$.
\end{proof}

\begin{proof}[Proof of \Cref{thm:sampled-cover}]
The output $\widehat m=\max\{0,D(C)-\widehat h\}$ satisfies
\[
  |\widehat m-m|
  \le|D(C)-\widehat h-m|
  \le r(C)+|\widehat h - h(C)|.
\]
Here the first inequality uses $m>0$, and the second uses $m=D(C)-h(C)+r(C)$.
By \Cref{lem:sampled-approx-cover,lem:internal-correction} and a union bound, both terms are at most $\eps m/2$ with probability at least $5/6$.
On every execution, the greedy matching loop costs $O(\nu/\eps)$ random edge queries and ends with $b\le2\nu$.
Computing $D(C)$ costs $|C|<2b$ degree queries, and estimating $h(C)$ costs $\lceil Kb/\eps\rceil$ random edge queries.
Thus the total query complexity is $O(\nu/\eps)$ on every execution.
\end{proof}

\section{Omitted Details on Lower Bound Strategy}
In this section, we provide the formal reduction and standard facts regarding distances and divergences between probability distributions for the lower bound arguments presented in \Cref{sec:edge-lower,sec:moment-lower}.

\subsection{From Estimation to Testing}
\label{sec:lower_bound_formal}
For completeness, we record the standard reduction from parameter estimation to binary hypothesis testing (Le Cam's method; see Yu~\cite{Yu1997}), including a truncation argument to handle expected query complexity.
The following formulation applies to all our lower bound constructions.

\begin{lemma}
\label{lem:estimation-to-testing}
Let $\theta$ be a nonnegative graph parameter and $\mathsf H_0,\mathsf H_1$ be two graph distributions.
Fix $0<\eps<1$, $\delta,\gamma,\tau\in[0,1]$, and an integer $q\ge1$.
Suppose that under each $\mathsf H_i$, the value $\theta(G)$ belongs to a set $I_i\subseteq\mathbb R_{\ge0}$ with probability at least $1-\gamma$, and that
\[
  J_i:=\bigcup_{x\in I_i}[(1-\eps)x,(1+\eps)x],
  \qquad J_0\cap J_1=\varnothing.
\]
Suppose that for every algorithm making at most $q$ queries, the distributions $\mathsf P_0,\mathsf P_1$ of its observations and private randomness under the two hypotheses satisfy $\TV(\mathsf P_0,\mathsf P_1)\le\tau$.
Then any algorithm returning a $(1\pm\eps)$-estimate of $\theta(G)$ with probability at least $1-\delta$ on every graph in the two supports has worst case expected query complexity $T$ on those supports satisfying
\begin{equation}
  T\ge q\left(\frac{1-\tau}{2}-\delta-\gamma\right).
  \label{eq:estimation-testing-cost}
\end{equation}
In particular, $T=\Omega(q)$ whenever the expression in parentheses is bounded below by a positive constant.
\end{lemma}

\begin{proof}
Let $Q$ be the estimator's query count (define $Q=\infty$ if the estimator does not terminate), and assume $T<\infty$ since otherwise the claim is immediate.
Under either hypothesis $\mathsf{H}_i$, the estimate lies in the correct set $J_i$ except with probability at most $\delta+\gamma$.
Thus, declaring hypothesis $i$ when the estimate lies in $J_i$ and deciding arbitrarily otherwise gives a test with error probability at most $\delta+\gamma$.
We truncate this test before its $(q+1)$-th query and decide arbitrarily upon truncation.
Because truncation alters the decision only when $Q > q$, Markov's inequality implies that truncation changes the decision with probability at most $\Prob_{\mathsf{H}_i}[Q>q]\le T/q$ under each hypothesis.

For the truncated test, let $A$ be the event that it declares hypothesis $0$.
Under a uniform prior, its success probability is
\[
  \frac{1+\mathsf P_0[A]-\mathsf P_1[A]}2\le\frac{1+\TV(\mathsf P_0, \mathsf P_1)}2\le\frac{1+\tau}2.
\]
Combining this with the lower bound $1-\delta-\gamma-T/q$ on its success probability and rearranging proves \eqref{eq:estimation-testing-cost}.
\end{proof}

Throughout our applications, $\theta(G)$ is $m$ or $M_s$ and we consider $\delta=1/3$, $\gamma\le1/100$, and $\tau\le1/4$.
These parameter choices yield $\frac{1-\tau}{2}-\delta-\gamma \ge \frac{3}{8} - \frac{1}{3} - \frac{1}{100} = \frac{19}{600} > 0$, making the coefficient in \eqref{eq:estimation-testing-cost} a positive constant.
Consequently, the separation and the total variation bound together imply the claimed $\Omega(q)$ expected query lower bound.
We note that in each edge count construction, both hypotheses have the same vertex count $n$ (after padding with isolated vertices, following the same convention as Beretta and T\v{e}tek~\cite{BerettaTetek2024}); in each moment construction, they have the same $n$ and $m$.
Thus, supplying $k,n$ for edge estimation or $k,n,m$ for moment estimation to the algorithm leaves the lower bounds unchanged.

\subsection{Distances and Divergences}
For convenience, we collect the standard facts used in our lower bound proofs.
As our constructions are on discrete domains, we state the definitions and identities below for discrete probability distributions on a countable domain.
All logarithms are taken in the natural base.

The total variation distance and Kullback--Leibler (KL) divergence between probability distributions $P$ and $Q$ are defined as
\[
  \TV(P, Q) := \frac{1}{2}\sum_x |P(x) - Q(x)|,
  \qquad
  \KL(P \| Q) := \sum_{x: P(x)>0} P(x) \log\frac{P(x)}{Q(x)},
\]
with the convention that $\KL(P \| Q) = +\infty$ unless $P$ is absolutely continuous with respect to $Q$ (denoted $P \ll Q$, meaning $P(x) = 0$ whenever $Q(x) = 0$).
These quantities satisfy the following standard properties:
\begin{itemize}
  \item \textbf{Pinsker's inequality:}
  Total variation distance is bounded by KL divergence via
  \[
    \TV(P, Q) \le \sqrt{\frac{1}{2} \KL(P \| Q)}.
  \]
  \item \textbf{Chain rule for KL divergence:}
  For joint distributions $P_{X,Y}$ and $Q_{X,Y}$ with $P_X \ll Q_X$,
  \[
    \KL(P_{X,Y} \| Q_{X,Y}) = \KL(P_X \| Q_X) + \sum_{x: P_X(x)>0} P_X(x) \, \KL(P_{Y \mid X=x} \| Q_{Y \mid X=x}).
  \]
  In particular, when $P_X = Q_X$, the joint divergence equals the expected conditional divergence.
  \item \textbf{Subadditivity on product distributions:}
  For independent coordinates under both distributions, $\TV(P_X \otimes P_Y, \, Q_X \otimes Q_Y) \le \TV(P_X, Q_X) + \TV(P_Y, Q_Y)$.
  \item \textbf{Data processing inequality:}
  Any (possibly randomized) postprocessing of observations cannot increase distance or divergence.
  Formally, if an algorithm applied to inputs drawn from $P$ or $Q$ induces output distributions $P', Q'$, then $\TV(P', Q') \le \TV(P, Q)$ and $\KL(P' \| Q') \le \KL(P \| Q)$.
  In our lower bounds, this ensures that the distinguishing advantage of any adaptive algorithm cannot exceed the total variation distance between the responses of the augmented oracle that simulate it, and that summarizing query sequences or discarding intermediate observations cannot increase divergence.
\end{itemize}

\section{Analysis of Degree-Preserving Edge Splits}
\label{sec:lower_bound_edge_choice}
In this section, we provide the complete proofs and detailed analyses for the degree-preserving edge splits lower bound construction presented in \Cref{sec:lb_edge_1}.

\subparagraph*{Construction of degree-preserving edge splits.}
For completeness, we restate the construction as follows.
Let $L := \lceil 1/\eps \rceil$. We take a vertex set $U$ of $k$ vertices partitioned into two equal sets of size $k/2$, and place a $k/4$-regular bipartite graph $B$ between them, so $|E(B)| = k^2/8$.
To each vertex $u \in U$, we attach $L$ leaf vertices of its own; these $kL$ leaf edges are present in every graph.

For each edge $e = \{u,v\} \in E(B)$, we introduce two auxiliary vertices $x_e, y_e$ and form a \emph{gadget} on $\{u, v, x_e, y_e\}$ by selecting one of two local configurations:
in the \emph{single} configuration, the edge $\{u,v\}$ is present while $x_e, y_e$ remain isolated;
in the \emph{split} configuration, the two edges $\{u,x_e\}$ and $\{v,y_e\}$ are present while $\{u,v\}$ is absent.
We refer to the edges present in these configurations as \emph{gadget} edges, distinguishing them from the $kL$ leaf edges.
Both configurations contribute one incident edge to each of $u$ and $v$, so every vertex $u \in U$ has fixed degree $d_u = L + k/4$ regardless of gadget choices.
We define two graph distributions as follows:
\begin{description}
    \item[$\mathsf{H}_0$] Independently assign each $e \in E(B)$ to be \emph{single} or \emph{split} with probability $1/2$ each.
    \item[$\mathsf{H}_1$] Bias the probability of \emph{single} to $1/2 + \beta$, where $\beta:=C_0\eps(L+k)/k$ for a sufficiently large constant $C_0$.
\end{description}
With $C_0$ fixed, let $k_0$ and $\eps_0$ be constants such that $\beta \le C_0(2/k+\eps) \le 1/4$ whenever $k \ge k_0$ and $0 < \eps \le \eps_0$.

Let $X$ count the gadgets in the single configuration, so $m(G)=kL+k^2/4-X$.
Under $\mathsf H_0$ and $\mathsf H_1$, $X$ is binomial with $k^2/8$ trials and success probabilities $1/2$ and $1/2+\beta$, respectively.
\begin{lemma}
\label{lem:switch-separation}
For $k\ge k_0$ a multiple of four and $0<\eps\le\eps_0$, every graph in either support has matching number $\nu(G) = k$.
Furthermore, each $\mathsf H_i$ places probability at least $99/100$ on a set $I_i$ of edge counts, and $I_0$ and $I_1$ are $(1\pm\eps)$-separated.
\end{lemma}

\begin{proof}
Every edge in any graph $G$ from either support has at least one endpoint in $U$, so $U$ forms a vertex cover of size $k$. 
Conversely, selecting one leaf edge for each of the $k$ vertices in $U$ yields a matching of size $k$, ensuring $\nu(G) = k$.

Write $\mu_i=\E_{\mathsf H_i}[m]$, so that $\mu_0-\mu_1=\beta k^2/8$. Let $I_i$ be the set of attainable edge counts within $(\mu_0-\mu_1)/4$ of $\mu_i$.
Under either hypothesis, $\Var[m]\le k^2/32$, so for sufficiently large $C_0$, Chebyshev's inequality gives $\Prob[m\notin I_i]=O(k^2/(\mu_0-\mu_1)^2)=O(C_0^{-2})\le1/100$ using $\eps(L+k)\ge1$.
For $a\in I_0$ and $b\in I_1$, we have $a,b\le k(L+k)$ and
\[
  a-b\ge\frac{\mu_0-\mu_1}{2}
  =\frac{C_0}{16}\eps k(L+k)
  >\eps(a+b),
\]
again for sufficiently large $C_0$.
Thus $(1-\eps)a>(1+\eps)b$, proving the separation.
\end{proof}

\subparagraph*{Indistinguishability of degree-preserving edge splits.}
Since every vertex in $U$ has degree $L + k/4$ and all other nonisolated vertices have degree $1$, endpoint degrees provide no distinguishing information beyond the sampled edge identities. It therefore suffices to analyze the distribution of the sampled edges alone.

Because $m$ depends on the configurations of all gadgets, the sampling probabilities of individual edges depend on them as well.
To remove this dependence, we consider an alternative sampling process of obtaining random edges via rejection sampling from a fixed candidate universe.

Let $F$ denote the candidate universe comprising the $kL$ leaf edges together with, for each edge $e=\{u,v\}\in E(B)$, the three candidate edges $\{u,v\}$, $\{u,x_e\}$, $\{v,y_e\}$, each of which may or may not be present in $G$. 
The size of the candidate universe $\abs{F}$ is identical under both hypotheses as $\abs{F} = kL + 3\abs{E(B)} = kL + \frac{3k^2}{8}$, and only the \emph{presence} of these candidates in $G$ depends on the graph $G \sim \mathsf{H}_i$.
To sample $q \ge 1$ edges, draw candidates independently and uniformly at random from $F$ with replacement. For each draw, record the candidate's identity and an indicator specifying whether it is present in $G$.
A candidate is accepted if and only if it is present in $G$, and sampling terminates upon reaching $q$ acceptances (counting repeated edges with multiplicity).

\begin{lemma}
\label{lem:switch-rejection-cost}
Let $G$ be a graph sampled from either $\mathsf H_0$ or $\mathsf H_1$, and $q \geq 1$ be an integer.
\begin{itemize}
  \item Conditional on $G$, the $q$ accepted edges obtained by the rejection sampling process above are independent and uniform in $E(G)$.
  \item The total number of candidate draws $T$ required to obtain $q$ acceptances satisfies $T < \infty$ almost surely and $\E[T \mid G] \le 3q$.
\end{itemize}
\end{lemma}

\begin{proof}
Fix $G$ throughout the proof and write $m(G) := \abs{E(G)}$. 
Since $G$ contains all $kL$ leaf edges and at least one edge for each $e \in E(B)$, we have
\[
  m(G)\ge kL+\frac{k^2}{8}
  \qquad\text{and}\qquad
  \frac{m(G)}{|F|}
  \ge\frac{kL+k^2/8}{kL+3k^2/8}\ge\frac13.
\]
Each candidate draw is accepted independently with probability $m(G)/|F| \ge 1/3$.
The total number of candidate draws $T$ thus follows a negative binomial distribution with mean $\E[T] = q \abs{F} / m(G) \le 3q$, guaranteeing $T < \infty$ almost surely.

For any edge $f \in E(G)$, the probability that $f$ is chosen on a given acceptance is
\[
  \Prob[\text{accepted edge} = f] = \sum_{r=0}^{\infty} \left(1 - \frac{m(G)}{|F|}\right)^r \frac{1}{|F|} = \frac{1}{m(G)}.
\]
After each acceptance, the remaining candidates are again independent uniform draws from $F$. 
Hence the accepted edges are independent and uniform in $E(G)$. 
\end{proof}

We now bound the total variation distance between the distributions of the responses under the two hypotheses. We allow the algorithm access to the fixed vertex labels from the construction; this only strengthens the lower bound, as labeled vertices reveal at least as much information as opaque identifiers.
\begin{lemma}
\label{lem:switch-fixed-sample-tv}
Fix an integer $q \ge 1$, and let $\mathsf P_i$ denote the joint distribution of the responses to $q$ queries to the augmented oracle of \Cref{sec:edge-lower} under $\mathsf H_i$.
Then
\[
  \KL(\mathsf P_0 \| \mathsf P_1)
  \le 8q\beta^2 \frac{k}{L+k}
  = 8C_0^2 q \eps^2 \frac{L+k}{k},
\]
which implies by Pinsker's inequality that
\[
  \TV(\mathsf P_0, \mathsf P_1)
  \le 2\beta \sqrt{\frac{qk}{L+k}}
  = 2C_0 \sqrt{\frac{q\eps^2(L+k)}{k}}.
\]
In particular, for $c := 1/(64 C_0^2)$, every integer $1 \le q \le c \frac{k}{\eps^2(L+k)}$ satisfies $\TV(\mathsf P_0, \mathsf P_1) \le 1/4$.
\end{lemma}

\begin{proof}
We consider the rejection sampling process of \Cref{lem:switch-rejection-cost}.
Let $T$ be as in \Cref{lem:switch-rejection-cost}, and let $\mathsf Q_i$ be the joint distribution of the candidate identities and acceptance indicators observed through time $T$ under $\mathsf H_i$.
Because the accepted edges determine the endpoint degrees, the $q$ responses are a function of the candidate identities and acceptance indicators observed through time $T$ under both hypotheses.
The data processing inequality therefore gives $\KL(\mathsf P_0 \| \mathsf P_1) \le \KL(\mathsf Q_0 \| \mathsf Q_1)$.

To bound $\KL(\mathsf Q_0 \| \mathsf Q_1)$, we analyze the information revealed at each step $t \le T$ conditional on the history.
Every candidate is drawn uniformly at random from $F$ independently of the graph, so the identity of the candidate drawn at step $t$ contributes zero divergence.
Any distinguishing information thus must come from its acceptance indicator:
\begin{itemize}
  \item If the candidate is a leaf edge, it is present in $G$ with probability $1$ under both hypotheses, contributing zero divergence.
  \item If the candidate belongs to a gadget $e \in E(B)$ one of whose candidates was drawn at an earlier step, the configuration of $e$ is already determined by the history, again contributing zero divergence.
  \item If no candidate of its gadget $e \in E(B)$ was drawn at an earlier step, the configuration of $e$ still has its original Bernoulli distribution under each hypothesis.
  Observing whether the candidate is present distinguishes $\operatorname{Ber}(1/2)$ from $\operatorname{Ber}(1/2\pm\beta)$, which yields a KL divergence of
  \[
    \KL\left(\operatorname{Ber}\left(\frac12\right) \middle\| \operatorname{Ber}\left(\frac12\pm\beta\right)\right)
    = -\frac12 \log(1-4\beta^2) 
    \le \frac{2\beta^2}{1-4\beta^2}
    \le \frac83 \beta^2,
  \]
  where the first inequality uses $-\log(1-x) \le x/(1-x)$ for $0 \le x < 1$ and the second uses $\beta \le 1/4$.
\end{itemize}
Because gadget candidates comprise $3|E(B)| / |F| = \frac{3k^2/8}{kL + 3k^2/8} = \frac{3k}{8L + 3k} \le \frac{k}{L+k}$ of the universe $F$, the expected divergence contributed by each draw is at most $\frac{8}{3}\beta^2 \cdot \frac{3|E(B)|}{|F|} \le \frac{8}{3}\beta^2 \frac{k}{L+k}$.

Since no candidates are drawn after time $T$, extend both processes by a fixed dummy symbol at every step $t > T$; these steps contribute zero divergence, and the observations through time $T$ are a function of the extended sequence.
Because the event $\{T \ge t\}$ is determined by the history before step $t$, the chain rule bounds the divergence of the first $t_0$ steps of the extended sequence by $\frac{8}{3}\beta^2 \frac{k}{L+k} \sum_{t=1}^{t_0} \Prob_{\mathsf H_0}[T \ge t]$.
Letting $t_0 \to \infty$ and applying the data processing inequality, we obtain
\[
  \KL(\mathsf Q_0 \| \mathsf Q_1)
  \le \left(\frac{8}{3}\beta^2 \frac{k}{L+k}\right) \sum_{t=1}^{\infty} \Prob_{\mathsf H_0}[T \ge t]
  = \left(\frac{8}{3}\beta^2 \frac{k}{L+k}\right) \E_{\mathsf H_0}[T]
  \le 8q\beta^2 \frac{k}{L+k},
\]
where the final inequality follows from $\E_{\mathsf H_0}[T] \le 3q$ (\Cref{lem:switch-rejection-cost}).
The data processing inequality and Pinsker's inequality then yield, for every integer $1 \le q \le ck/(\eps^2(L+k))$ with $c = 1/(64C_0^2)$,
\[
  \TV(\mathsf P_0, \mathsf P_1)
  \le \sqrt{\frac{1}{2} \KL(\mathsf P_0 \| \mathsf P_1)}
  \le 2\beta \sqrt{\frac{qk}{L+k}}
  = 2C_0 \sqrt{\frac{q\eps^2(L+k)}{k}}
  \le 2C_0\sqrt{c} \le \frac{1}{4},
\]
completing the proof.
\end{proof}
Set $q := \left\lfloor c \frac{k}{\eps^2(L+k)} \right\rfloor$.
Since $L = \lceil 1/\eps \rceil$, choosing $\eps_0$ sufficiently small after fixing $c$ ensures $q \ge 1$ and $q = \Theta(\min\{k/\eps, 1/\eps^2\})$.
By \Cref{lem:switch-fixed-sample-tv}, this choice of $q$ ensures $\TV(\mathsf P_0, \mathsf P_1) \le 1/4$.
Applying \Cref{lem:estimation-to-testing} alongside the separation in \Cref{lem:switch-separation} establishes the worst case expected query lower bound
\[
  \Omega(q)
  = \Omega\!\left(\frac{k}{\eps^2(L+k)}\right)
  = \Omega\!\left(\min\!\left\{\frac{k}{\eps}, \frac{1}{\eps^2}\right\}\right)
\]
for every multiple of four $k \ge k_0$.

\subparagraph*{Removing the assumptions on $k$.}
It remains to remove the assumption that $k \ge k_0$ is a multiple of four.
\begin{description}
  \item[Not a multiple of four ($k \ge k_0$):] When $k \ge k_0$ is not a multiple of four, we set $k' := 4 \lfloor k/4 \rfloor$; taking $k_0$ to be a multiple of four ensures $k' \ge k_0$. Applying the construction above on $k'$ produces graphs with matching number $\nu(G) = k' \le k$. Because $k' \ge k/2$, rounding down to a multiple of four alters the lower bound by at most a constant factor.
  \item[Bounded matching number ($2 \le k \le k_0$):] We give a separate construction for this regime.
  Note that as $k$ is bounded by a constant, the target rate simplifies to $\Omega(\min\{k/\eps, 1/\eps^2\}) = \Omega(1/\eps)$.
  Set $D := \lfloor 1/(4\eps) \rfloor$, and choose $\eps_0 > 0$ sufficiently small that $D \ge 12$.
  We define $\mathsf{H}_0$ as a point mass on a star $K_{1, D}$ ($m = D$), and $\mathsf{H}_1$ as a point mass on the graph obtained by inserting an additional edge between two designated leaves of $K_{1,D}$ ($m = D + 1$). Both graphs have matching number at most $2 \le k$. Because $\eps(2D + 1) < 1$, we have $(1+\eps)D < (1-\eps)(D+1)$, so the edge counts $D$ and $D+1$ are $(1\pm\eps)$-separated.

  Consider the augmented oracle returning a uniformly random edge together with both endpoint degrees.
  If $G \sim \mathsf{H}_0$, every query returns a star edge with endpoint degrees $(D, 1)$, each drawn with probability $1/D$.
  If $G \sim \mathsf{H}_1$, the two designated leaves have degree $2$, so any query hitting them or the extra edge returns endpoint degrees $(D, 2)$ or $(2, 2)$, which never occurs under $\mathsf{H}_0$.
  The $D-2$ unchanged star edges are the only responses shared between the two distributions, carrying a total shared probability mass of $(D-2)/(D+1)$, so a single response has total variation distance $1 - \frac{D-2}{D+1} = \frac{3}{D+1}$ between the two hypotheses.
  By subadditivity of total variation distance over independent responses, $q := \lfloor (D+1)/12 \rfloor$ queries give $\TV(\mathsf{P}_0, \mathsf{P}_1) \le \frac{3q}{D+1} \le \frac{1}{4}$.
  Applying \Cref{lem:estimation-to-testing} yields the worst case expected query lower bound $\Omega(q) = \Omega(1/\eps) = \Omega(k/\eps)$, which completes the proof of the lower bound $\Omega(\min\{k/\eps, 1/\eps^2\})$ for all $k \ge 2$.
\end{description}

\section{Analysis of Random Matching Subset}
\label{sec:lower_bound_random_subset_matching}
In this section, we provide the complete proofs and detailed analyses for the random matching subset lower bound construction presented in \Cref{sec:lb_edge_2}.

\subparagraph*{Construction of random matching subset.}
We adapt the weighted sampling hard family of Beretta and T\v{e}tek~\cite{BerettaTetek2024} to graphs by realizing each unit weight element as an isolated edge.
Let $M$ be a matching of $k$ isolated edges on $2k$ vertices.
Throughout, assume $k \ge \eps^{-2}$ and $0 < \eps \le \eps_0 \le 1/8$, and let $j := \lfloor(1-3\eps)k\rfloor$, so that $1 \le j \le k$.
Under $\mathsf H_0$, the graph is the full matching $G=M$.
Under $\mathsf H_1$, the graph is the subgraph of $M$ formed by a uniformly random $j$-subset of its edges, leaving the remaining $2(k - j)$ vertices isolated.
Every graph in either support has matching number at most $k$, and every nonisolated vertex has degree $1$.
Moreover, $(1+\eps)j\le(1+\eps)(1-3\eps)k<(1-\eps)k$, so the edge counts $k$ and $j$ are $(1\pm\eps)$-separated.

\subparagraph*{Indistinguishability of random matching subset.}
Let $\mathsf P_i$ denote the joint distribution of the responses to $q$ queries to the augmented oracle of \Cref{sec:edge-lower} under $\mathsf H_i$.
Because every nonisolated vertex has degree $1$ under both hypotheses, endpoint degrees provide no distinguishing information, so it suffices to analyze the sampled edge identities.
\begin{lemma}
\label{lem:hidden-support-likelihood}
Fix an integer $q\ge1$ and let $S$ be the sequence of $q$ sampled edges.
For any sequence $s = (e_1, \ldots, e_q)$ containing $r \le j$ distinct edges,
\[
  \frac{\Prob_{\mathsf{H}_1}[S = s]}{\Prob_{\mathsf{H}_0}[S = s]}
  = \left(\frac{k}{j}\right)^q \frac{\binom{k-r}{j-r}}{\binom{k}{j}}.
\]
This ratio is zero when $r > j$.
In particular, for fixed $q$, the likelihood ratio depends solely on the number $r$ of distinct observed edges.
\end{lemma}

\begin{proof}
Under $\mathsf{H}_0$, the graph is the full matching $G = M$, so every sequence $s$ has probability $\Prob_{\mathsf{H}_0}[S = s] = k^{-q}$.
Under $\mathsf{H}_1$, $\Prob_{\mathsf{H}_1}[S = s] = 0$ for $r > j$ since a graph with $j$ edges cannot produce more than $j$ distinct sampled edges.
For $r \le j$, $\Prob[S = s \mid G]$ equals $j^{-q}$ if $G$ contains all $r$ distinct edges of $s$, and zero otherwise.
The graph $G$ is chosen uniformly at random among all $\binom{k}{j}$ $j$-subsets of $M$ under $\mathsf{H}_1$, of which $\binom{k-r}{j-r}$ contain the $r$ distinct edges of $s$.
Averaging over $G \sim \mathsf{H}_1$ gives $\Prob_{\mathsf{H}_1}[S = s] = \frac{\binom{k-r}{j-r}}{\binom{k}{j}} j^{-q}$. 
Dividing this by $\Prob_{\mathsf{H}_0}[S = s] = k^{-q}$ yields the stated ratio.
\end{proof}
Therefore the distribution of sampled edge sequences given $r$ is identical under both hypotheses, meaning that distinguishing the two distributions rests on the number of draws that repeat an earlier edge.
We now bound the KL divergence by considering whether each successive draw repeats an observed edge.
\begin{lemma}
\label{lem:hidden-support-kl}
There exists a constant $C>0$ such that for all $0<\eps\le1/8$, $k\ge\eps^{-2}$, and every integer $1\le q\le k/2$, we have $\KL\left(\mathsf P_1\middle\|\mathsf P_0\right)\le C\eps^2q^2 /k$.
In particular, for a sufficiently small constant $c \in (0, 1/2]$, every integer $1 \le q \le c\sqrt{k}/\eps$ satisfies $\TV(\mathsf P_0, \mathsf P_1) \le 1/4$.
\end{lemma}

\begin{proof}
Since $k\ge\eps^{-2}$ and $\eps\le1/8$, we have $k-j \le 3\eps k + 1 \le 4\eps k \le k/2$ and $j\ge k/2$.
Fix a history of $t < q$ queries containing $r$ distinct edges.
Under $\mathsf H_1$, the graph remains uniform among all $j$-subsets containing these $r$ edges.
We classify the $(t+1)$-th response into two mutually exclusive events:
\begin{itemize}
  \item \textbf{Repeating an observed edge:}  A repeat occurs with probability $r / k$ and $r / j$ under $\mathsf{H}_0$ and $\mathsf{H}_1$, respectively.
  Under both hypotheses, the repeated edge is uniform among the $r$ observed edges.
  \item \textbf{Observing a new edge:} A new edge is observed with probability $1 - r / k$ and $1 - r / j$ under $\mathsf{H}_0$ and $\mathsf{H}_1$, respectively.
  Under both hypotheses, the new edge is uniform among the $k - r$ unobserved edges.
\end{itemize}
Since the two hypotheses differ only in the probability of a repeat, the divergence at step $t+1$ reduces to that of the repeat indicator $\KL(\operatorname{Ber}(r/j) \| \operatorname{Ber}(r/k))$. 
For $r = 0$, this divergence vanishes.
For $r \ge 1$, applying $\log x \le x - 1$ yields
\[
  \KL\left(\operatorname{Ber}(r/j)\middle\|\operatorname{Ber}(r/k)\right)
  \le\frac{r(k-j)^2}{j^2(k-r)}
  = O\!\left(\frac{\eps^2 t}{k}\right),
\]
using $j, k-r \ge k/2$, $k-j \le 4\eps k$, and $r \le t$.
Summing over all steps $t = 0, \ldots, q-1$ by the chain rule for KL divergence gives $\KL\left(\mathsf P_1\middle\|\mathsf P_0\right)\le C\eps^2q^2 /k$ for a sufficiently large constant $C > 0$.
By Pinsker's inequality, $\TV(\mathsf P_0, \mathsf P_1) \le \sqrt{\frac{1}{2} \KL(\mathsf P_1 \| \mathsf P_0)} \le \sqrt{C\eps^2 q^2 / (2k)} \le c\sqrt{C/2} \le 1/4$ for a sufficiently small constant $c \in (0, 1/2]$.
\end{proof}
Set $q := \lfloor c\sqrt{k}/\eps \rfloor$.
Since $k \ge \eps^{-2}$, we have $\eps^{-2} \le \sqrt{k}/\eps \le k$, so $c \le 1/2$ gives $q \le k/2$, and choosing $\eps_0 \le 1/8$ sufficiently small after fixing $c$ ensures $q \ge 1$ and $q = \Theta(\sqrt{k}/\eps)$.
By \Cref{lem:hidden-support-kl}, this choice of $q$ ensures $\TV(\mathsf P_0, \mathsf P_1) \le 1/4$.
Applying \Cref{lem:estimation-to-testing} alongside the separation of $k$ and $j$ establishes the worst case expected query lower bound $\Omega(q) = \Omega(\sqrt{k}/\eps)$ for the regime $k \ge \eps^{-2}$, completing the proof of \Cref{thm:intro-edge-lb}.

\section{Analysis of Algorithm~\ref{alg:moment-estimation}}
\label{sec:moment-detail}
In this section, we provide the omitted analysis of \Cref{alg:moment-estimation}, completing the proof of \Cref{thm:intro-moment}.
Throughout this section, fix an accuracy parameter $0 < \eps \le 1$ and an integer $s \ge 2$, and write $W := 2m$, $M_s := \sum_{v \in V} (d_v)^s$, and $\Lambda_s := W / M_s^{1/s}$.

\subsection{The Moment Estimator}
We first analyze the mean and variance of the Horvitz--Thompson estimator $Z_q(W)$ assuming the total degree sum $W$ is known.
Given an expected sample size $q > 0$, the algorithm draws $N \sim \operatorname{Pois}(q)$ and takes $N$ independent fair endpoint samples.
Let $N_v$ count the occurrences of vertex $v \in V$, $I_v := \Ind[N_v > 0]$ indicate whether $v$ is sampled at least once, and $S := \{v \in V : I_v = 1\}$ denote the set of distinct observed vertices.
We define the estimator as
\[
  Z_q(W) := \sum_{v \in S} \frac{(d_v)^s}{\pi_v(q)}
  = \sum_{v \in V} I_v \frac{(d_v)^s}{\pi_v(q)},
  \qquad
  \pi_v(q) := 1 - e^{-q d_v / W},
\]
with an empty sum evaluating to zero.
\begin{lemma}
\label{lem:moment-ht-props}
For every $q > 0$, the mean and variance of the estimator $Z_q(W)$ are
\[
  \E[Z_q(W)] = M_s,
  \qquad
  \Var[Z_q(W)] = \sum_{v \in V} \frac{(d_v)^{2s}}{e^{q d_v / W} - 1}.
\]
\end{lemma}

\begin{proof}
Because $N \sim \operatorname{Pois}(q)$ and each of the $N$ samples independently selects vertex $v$ with probability $d_v / W$, Poisson splitting ensures that the $N_v$ are mutually independent with $N_v \sim \operatorname{Pois}(q d_v / W)$; hence the indicators $I_v$ are also mutually independent with $\Prob[I_v = 1] = 1 - \Prob[N_v = 0] = \pi_v(q)$.
By linearity of expectation,
\[
  \E[Z_q(W)]
  =\sum_{v\in V}\frac{(d_v)^s\E[I_v]}{\pi_v(q)}
  =M_s.
\]
Because the indicators $I_v$ are mutually independent, the variance is
\[
  \Var[Z_q(W)]
  =\sum_{v\in V}\frac{(d_v)^{2s}\Var[I_v]}{\pi_v(q)^2}
  =\sum_{v\in V}(d_v)^{2s}\frac{1-\pi_v(q)}{\pi_v(q)}
  =\sum_{v\in V}\frac{(d_v)^{2s}}{e^{qd_v/W}-1}. \qedhere
\]
\end{proof}
We now bound this variance in terms of $\Lambda_s$, and $\Lambda_s$ in terms of $\nu$.
For $q > 0$, write $R_s(q) := \Var[Z_q(W)] / M_s^2$.

\begin{lemma}
\label{lem:moment-ht-relvar}
For every $q>0$,
\[
  R_s(q)
  \le\min\!\left\{
    s!\left(\frac{\Lambda_s}{q}\right)^s,\,
    (2s-1)!\frac{\Lambda_s^{2s}}{q^{2s-1}}
  \right\}.
\]
Furthermore, $\Lambda_s\le2(2\nu)^{1-1/s}$.
\end{lemma}

\begin{proof}
For every integer $a \ge 1$ and $x \ge 0$, the exponential series yields $e^x - 1 \ge x^a / a!$.
Applying this with $x = q d_v / W$ for $a = s$ gives
\[
  \Var[Z_q(W)]
  \le s!\left(\frac Wq\right)^s\sum_{v\in V}(d_v)^s
  =s!\left(\frac Wq\right)^sM_s.
\]
Similarly, taking $a = 2s - 1$ and using $\sum_{v \in V} d_v = W$ yields
\[
  \Var[Z_q(W)]
  \le(2s-1)!\left(\frac Wq\right)^{2s-1}\sum_{v\in V}d_v
  =(2s-1)!\frac{W^{2s}}{q^{2s-1}}.
\]
Dividing both bounds by $M_s^2$ and recalling $\Lambda_s = W / M_s^{1/s}$ proves the variance bound.

For the bound on $\Lambda_s$, let $C^\star$ be the endpoints of a maximum matching. Since $\sum_{v\in C^\star}d_v\ge W/2$, H\"older's inequality gives
\[
  \frac W2 \le\sum_{v\in C^\star}d_v \le |C^\star|^{1-1/s} \left(\sum_{v\in C^\star}(d_v)^s\right)^{1/s} \le(2\nu)^{1-1/s}M_s^{1/s}.
\]
Dividing by $M_s^{1/s}$ completes the proof.
\end{proof}

\subsection{Choosing the Parameters Without Advice}
We evaluate the estimator with an arbitrary estimate $U > 0$ in place of $W$ by setting
\[
  Z_q(U) := \sum_{v \in S} \frac{(d_v)^s}{\pi_v(q, U)} = \sum_{v \in V} I_v \frac{(d_v)^s}{\pi_v(q, U)}, \qquad \pi_v(q, U) := 1 - e^{-q d_v / U},
\]
with an empty sum evaluating to zero.
Replacing $W$ with $U$ alters only the denominators while each vertex is still observed with probability $\pi_v(q)$ as the samples do not depend on $U$.
Recall that $R_s(q) = \Var[Z_q(W)] / M_s^2$. 
\begin{lemma}
\label{lem:moment-estimated-total-degree}
Fix $q, U > 0$ and write $c := U / W$.
For every $v \in V$,
\begin{equation}
  \min\{1, c\} \le \frac{\pi_v(q)}{\pi_v(q, U)} \le \max\{1, c\}.
  \label{eq:moment-inclusion-ratio}
\end{equation}
Consequently, 
\[
\min\{1, c\} M_s \le \E[Z_q(U)] \le \max\{1, c\} M_s, \qquad \Var[Z_q(U)] \le \max\{1, c\}^2 \Var[Z_q(W)].
\]
Furthermore, if $1/2 \le c \le 2$, then
\[
  \frac{\abs{\E[Z_q(U)] - M_s}}{M_s} \le 4|c - 1| \min\left\{1, \sqrt{q R_s(q/2)}\right\}, \quad \Var[Z_q(U)] \le 4 M_s^2 R_s(q).
\]
\end{lemma}

\begin{proof}
Because $t \mapsto 1 - e^{-t}$ is strictly increasing and concave with value $0$ at the origin, the slope $(1 - e^{-t})/t$ is strictly decreasing on $(0, \infty)$.
For $x_v := q d_v / W > 0$, the monotonicity of $1 - e^{-t}$ and of $(1 - e^{-t})/t$ implies that $\frac{\pi_v(q)}{\pi_v(q, U)}$ lies in $[1, c]$ when $c \ge 1$, and in $[c, 1]$ when $0 < c \le 1$, establishing \eqref{eq:moment-inclusion-ratio}.
Since $\E[Z_q(U)] = \sum_{v \in V} (d_v)^s \pi_v(q)/\pi_v(q, U)$, \eqref{eq:moment-inclusion-ratio} gives the bounds on the mean.
Independence of the indicators $I_v$ then yields
\begin{align*}
  \Var[Z_q(U)]
  &= \sum_{v \in V} (d_v)^{2s} \frac{\pi_v(q)(1 - \pi_v(q))}{\pi_v(q, U)^2}\\
  &\le \max\{1, c\}^2 \sum_{v \in V} \frac{(d_v)^{2s}(1 - \pi_v(q))}{\pi_v(q)}= \max\{1, c\}^2 \Var[Z_q(W)].
\end{align*}
When $1/2 \le c \le 2$, this shows $\Var[Z_q(U)] \le 4 M_s^2 R_s(q)$.

The bounds on the mean already give $|\E[Z_q(U)] - M_s| \le |c - 1| M_s$.
For the refined bound, the mean value theorem gives $|e^{-a} - e^{-b}| \le |a - b| e^{-\min\{a, b\}}$ for $a, b > 0$.
Applying this with $a = x_v$ and $b = x_v/c$, and using $1/2 \le c \le 2$, yields
\[
  \left|\frac{1 - e^{-x_v}}{1 - e^{-x_v/c}} - 1\right|
  \le \frac{|c - 1| x_v/c}{e^{\min\{x_v, x_v/c\}} - 1}
  \le 4 |c - 1| \frac{x_v/2}{e^{x_v/2} - 1}.
\]
The bias thus satisfies
\[
  |\E[Z_q(U)] - M_s|
  \le \sum_{v \in V} (d_v)^s \left|\frac{\pi_v(q)}{\pi_v(q, U)} - 1\right|
  \le 4 |c - 1| \sum_{v \in V} (d_v)^s \frac{x_v/2}{e^{x_v/2} - 1}.
\]
Cauchy--Schwarz gives
\begin{align*}
  \sum_{v \in V} (d_v)^s \frac{x_v/2}{e^{x_v/2} - 1}
    &\le \left(\sum_{v \in V} \frac{(d_v)^{2s}}{e^{x_v/2} - 1}\right)^{1/2} \left(\sum_{v \in V} \frac{(x_v/2)^2}{e^{x_v/2} - 1}\right)^{1/2} \\
    &\le \sqrt{\Var[Z_{q/2}(W)]} \left(\sum_{v \in V} \frac{x_v}{2}\right)^{1/2} \\
    &= M_s \sqrt{R_s(q/2)} \sqrt{\frac{q}{2}}
    = M_s \sqrt{\frac{q}{2} R_s(q/2)},
\end{align*}
where the second inequality uses $y / (e^y - 1) \le 1$ for all $y > 0$, and the third step uses $\sum_{v \in V} x_v = q \sum_{v \in V} d_v / W = q$.
Combining the two bounds establishes the bias bound.
\end{proof}

\subparagraph*{Choosing the expected sample size.}
Running \Cref{alg:endpoint-capture} with accuracy $1/4$ and setting $\widetilde{W}$ to twice its output gives $\Prob[3W/4 \le \widetilde{W} \le 5W/4] \ge 63/64$ at an expected query complexity of $O(1 + \sqrt{\nu})$.
In search iteration $j \ge 0$, the algorithm sets
\[
  \Gamma_j := (2\widetilde{W})^s 2^{-j}, \qquad A_j := \frac{\widetilde{W}}{\Gamma_j^{1/s}} = 2^{j/s-1},
\]
draws $N_j \sim \operatorname{Pois}(4s A_j)$, takes $N_j$ fair endpoint samples, and evaluates $Z_j := Z_{4s A_j}(\widetilde{W})$.
The search terminates at the first index $J$ satisfying $Z_J \ge \Gamma_J / 2$ (with $J = \infty$ if it never stops).

\begin{lemma}
\label{lem:moment-search}
In the above process, $J$ is finite almost surely, and the following hold:
\begin{enumerate}
  \item Let $\mathcal E_0 := \{3W/4 \le \widetilde{W} \le 5W/4\}$.
  For a sufficiently large constant $K$, we have $\Prob[\Gamma_J > K M_s \mid \mathcal E_0] \le 1/64$.
  \item The selected value $A_J$ satisfies $\E[A_J] = O(\Lambda_s)$ and $\E[A_J^{2s/(2s-1)}] = O(\Lambda_s^{2s/(2s-1)})$.
  \item The expected total number of queries over all search iterations, excluding the initial call to \Cref{alg:endpoint-capture}, is $O(\Lambda_s)$.
\end{enumerate}
\end{lemma}

\begin{proof}
We first show that when $\widetilde{W}$ is accurate (event $\mathcal E_0$), the search rarely accepts an oversized candidate $\Gamma_j \gg M_s$ because the acceptance threshold $\Gamma_j / 2$ greatly exceeds the estimator's conditional expectation.
Fix the estimate $\widetilde{W} > 0$ and write $c := \widetilde{W} / W$.
For the analysis, define $Z_j$ from an independent sample batch for every $j$, including iterations not reached by the algorithm.
Conditional on $\widetilde{W}$, the event $\{J \ge j\}$ depends only on the outcomes of earlier iterations $0, \dots, j-1$ and is independent of $Z_j$.

Suppose first that $\mathcal E_0$ holds, so $3/4 \le c \le 5/4$.
By \Cref{lem:moment-estimated-total-degree}, $\E[Z_j \mid \widetilde{W}] \le 5M_s / 4$.
Because stopping at iteration $j$ requires $Z_j \ge \Gamma_j / 2$, Markov's inequality yields
\[
  \Prob[J = j \mid \widetilde{W}] \le \Prob[Z_j \ge \Gamma_j / 2 \mid \widetilde{W}] \le \frac{\E[Z_j \mid \widetilde{W}]}{\Gamma_j / 2} \le \frac{5M_s}{2\Gamma_j}.
\]
The candidate values $\Gamma_j = (2\widetilde{W})^s 2^{-j}$ halve at each iteration, so the terms $M_s / \Gamma_j$ double with $j$.
If $t$ is the largest index with $\Gamma_t > K M_s$ (with the sum understood as zero if no such index exists), then
\[
  \sum_{j:\,\Gamma_j > K M_s} \frac{M_s}{\Gamma_j} = \frac{M_s}{\Gamma_t} \sum_{r=0}^t 2^{-r} < \frac{2M_s}{\Gamma_t} < \frac{2}{K}.
\]
Summing over all such indices gives
\[
  \Prob[J < \infty, \, \Gamma_J > K M_s \mid \widetilde{W}] \le \frac{5}{2} \sum_{j:\,\Gamma_j > K M_s} \frac{M_s}{\Gamma_j} < \frac{5}{K},
\]
so choosing $K \ge 320$ makes this at most $1/64$ for every $\widetilde{W}$ in $\mathcal E_0$, and averaging over such $\widetilde{W}$ gives the first assertion, since $J < \infty$ almost surely, as shown below.

To bound the moments of $A_J$ and query complexity, we now allow arbitrary $c > 0$.
By \Cref{lem:moment-estimated-total-degree,lem:moment-ht-relvar}, using $s! \le s^s$, the expectation and variance satisfy
\[
  \E[Z_j \mid \widetilde{W}] \ge M_s \min\{1, c\}, \qquad
  \Var[Z_j \mid \widetilde{W}] \le 4^{-s} M_s^2 \left(\frac{\Lambda_s}{A_j}\right)^s \max\{1, c\}^2.
\]
To ensure that $Z_j \ge \Gamma_j / 2$ occurs with high probability, it suffices that $\Gamma_j \le M_s \min\{1, c\}$ (so that the acceptance threshold lies at least $M_s \min\{1, c\} / 2$ below the mean) and that the variance is small relative to $(M_s \min\{1, c\})^2$.
Since $\Gamma_j / M_s = (c \Lambda_s / A_j)^s$, the condition $\Gamma_j \le M_s \min\{1, c\}$ requires $A_j \ge c \Lambda_s$ when $c \ge 1$ and $A_j \ge c^{(s-1)/s} \Lambda_s$ when $c \le 1$.
Meanwhile, the variance bound gives $\Var[Z_j \mid \widetilde{W}] / (M_s \min\{1, c\})^2 \le 4^{-s} (\Lambda_s / A_j)^s \max\{c^2, c^{-2}\}$.
To make this upper bound at most $4^{-s}$, it suffices to take $A_j \ge c^{2/s} \Lambda_s$ when $c \ge 1$ and $A_j \ge c^{-2/s} \Lambda_s$ when $c \le 1$.
Because $s \ge 2$, both requirements are simultaneously satisfied when $A_j \ge \Lambda_s \max\{c, c^{-2/s}\}$.
Accordingly, define
\[
  j_\star := \min\{j \ge 0 : A_j \ge \Lambda_s \max\{c, c^{-2/s}\}\}, \qquad A_\star := A_{j_\star}.
\]
Because $M_s \le W^s$ and $\max\{c, c^{-2/s}\} \ge 1$, the threshold $\Lambda_s \max\{c, c^{-2/s}\}$ is at least $1$.
Since $A_0 = 1/2$ and $A_{j+1} = 2^{1/s} A_j$, the index $j_\star$ is positive and finite, and satisfies
\begin{equation}
  \Lambda_s \max\{c, c^{-2/s}\} \le A_\star < 2^{1/s} \Lambda_s \max\{c, c^{-2/s}\}.
  \label{eq:moment-search-crossing}
\end{equation}
For every $j \ge j_\star$, the candidate value satisfies
\[
  \frac{\Gamma_j}{M_s} = \left(\frac{c \Lambda_s}{A_j}\right)^s \le \left(\frac{c}{\max\{c, c^{-2/s}\}}\right)^s = \min\{1, c^{s+2}\} \le \min\{1, c\},
\]
so $\E[Z_j \mid \widetilde{W}] - \Gamma_j / 2 \ge M_s \min\{1, c\} / 2$.
Moreover, because $s \ge 2$,
\[
  \frac{\Var[Z_j \mid \widetilde{W}]}{(M_s \min\{1, c\})^2} \le 4^{-s} \left(\frac{\Lambda_s}{A_j}\right)^s \max\{c^2, c^{-2}\} \le 4^{-s} \frac{\max\{c^2, c^{-2}\}}{\max\{c, c^{-2/s}\}^s} \le 4^{-s}.
\]
Chebyshev's inequality therefore bounds the continuation probability at any reached iteration $j \ge j_\star$ by
\[
  \Prob[J > j \mid \widetilde{W}, J \ge j] = \Prob[Z_j < \Gamma_j / 2 \mid \widetilde{W}] \le \frac{4 \Var[Z_j \mid \widetilde{W}]}{(M_s \min\{1, c\})^2} \le 4^{1-s} \le \frac{1}{4}.
\]
Applying this bound successively at each reached iteration yields, for every integer $h \ge 0$,
\begin{equation}
  \Prob[J \ge j_\star + h \mid \widetilde{W}] \le 4^{-h}.
  \label{eq:moment-search-reach-tail}
\end{equation}
In particular, $\Prob[J = \infty \mid \widetilde{W}] \le \lim_{h \to \infty} 4^{-h} = 0$, so $J < \infty$ almost surely for every $\widetilde{W} > 0$.

Because the target exponents $1$ and $2s/(2s-1)$ are both at most $2$ for $s \ge 2$, it suffices to bound the second moment $\E[A_J^2]$ and apply Jensen's inequality.
When $J \le j_\star$, we have $A_J \le A_\star$.
For $h \ge 1$, iteration $j_\star + h$ has $A_{j_\star + h} = 2^{h/s} A_\star$ and, by \eqref{eq:moment-search-reach-tail}, is reached with probability at most $4^{-h}$.
Because $s \ge 2$, we have $2^{2/s} \le 2$, and hence
\[
  \E[A_J^2 \mid \widetilde{W}] \le A_\star^2 \left(1 + \sum_{h=1}^\infty 4^{-h} 2^{2h/s}\right) \le A_\star^2 \left(1 + \sum_{h=1}^\infty 2^{-h}\right) = 2 A_\star^2.
\]
By \eqref{eq:moment-search-crossing}, $A_\star^2 < 2^{2/s} \Lambda_s^2 \max\{c^2, c^{-4/s}\} \le 2 \Lambda_s^2 (c^2 + c^{-4/s})$.
Since $c = \widetilde{W}/W$ equals the ratio of the output of the first call to \Cref{alg:endpoint-capture} to $m$, the ratio moment bounds in \Cref{thm:endpoint-capture} give $\E[c^2] = O(1)$ and $\E[c^{-2}] = O(1)$.
As $2/s \le 1$, Jensen's inequality applied to $x \mapsto x^{2/s}$ gives
\[
  \E[c^{-4/s}] = \E\bigl[(c^{-2})^{2/s}\bigr] \le \bigl(\E[c^{-2}]\bigr)^{2/s} = O(1).
\]
Averaging over $\widetilde{W}$ therefore yields $\E[A_J^2] = O(\Lambda_s^2)$.
Applying Jensen's inequality to the concave functions $x \mapsto x^{1/2}$ and $x \mapsto x^{s/(2s-1)}$ gives
\[
  \E[A_J] \le \bigl(\E[A_J^2]\bigr)^{1/2} = O(\Lambda_s), \qquad
  \E[A_J^{2s/(2s-1)}] \le \bigl(\E[A_J^2]\bigr)^{s/(2s-1)} = O(\Lambda_s^{2s/(2s-1)}).
\]

Finally, let $Q_{\mathrm{search}}$ count the queries made during the parameter search.
In each reached iteration $j$, drawing $N_j \sim \operatorname{Pois}(4s A_j)$ fair endpoint samples and querying their degrees incurs an expected query cost at most $8s A_j$.
Because the search terminates at index $J$, exchanging summation and expectation yields
\[
  \E[Q_{\mathrm{search}}] \le 8s \sum_{j=0}^\infty A_j \Prob[J \ge j] = 8s \E\left[\sum_{j=0}^J A_j\right].
\]
Since $A_j = 2^{j/s-1}$ increases geometrically with ratio $2^{1/s} > 1$, the sum over all iterations up to $J$ is bounded by a constant multiple of $A_J$.
Taking expectations and using $\E[A_J]=O(\Lambda_s)$ gives $\E[Q_{\mathrm{search}}]=O(\Lambda_s)$, completing the proof.
\end{proof}

\subsection{Accuracy and Expected Query Complexity}
After the parameter search stops, \Cref{alg:moment-estimation} sets
\[
  q := D_s \left( 1 + \min\left\{ \frac{A_J}{\eps^{2/s}}, \, \frac{A_J^{2s/(2s-1)}}{\eps^{2/(2s-1)}} \right\} \right), \qquad \eta(q) := \kappa\max\left\{\eps, q^{-1/2}\right\},
\]
where $\kappa>0$ is a sufficiently small constant and $D_s\ge1$ is sufficiently large depending only on $s$.
The algorithm then runs \Cref{alg:endpoint-capture} with accuracy $\eta(q)$, lets $\widehat W$ be twice its output, and evaluates $Z_q(\widehat W)$.
These choices give a $(1\pm\eps)$-estimate of $M_s$ with probability at least $15/16$ and the expected query complexity claimed in \Cref{thm:intro-moment}.

\subparagraph*{Accuracy.}
Let $\mathcal E_0 := \{3W/4 \le \widetilde{W} \le 5W/4\}$ and $\mathcal E_1 := \{|\widehat{W}/W - 1| \le \eta(q)\}$.
By \Cref{thm:endpoint-capture}, $\Prob[\mathcal E_0^c] \le 1/64$.
For each selected $q$, the same theorem gives $\Prob[\mathcal E_1^c \mid q] \le 1/64$.
Because $q \ge 1$, choosing $\kappa \le 1/2$ ensures $\eta(q) \le 1/2$ and hence $1/2 \le \widehat{W}/W \le 2$ on $\mathcal E_1$.

Fix the constant $K$ from \Cref{lem:moment-search}, so $\Prob[\Gamma_J>KM_s\mid\mathcal E_0]\le1/64$.
On $\mathcal E_0\cap\{\Gamma_J\le KM_s\}$, we have $A_J=\widetilde W/\Gamma_J^{1/s}=\Omega(\Lambda_s)$.
Because each term in the definition of $q$ is increasing in $A_J$,
\[
  q = \Omega\left( D_s \min\left\{ \frac{\Lambda_s}{\eps^{2/s}}, \, \frac{\Lambda_s^{2s/(2s-1)}}{\eps^{2/(2s-1)}} \right\} \right),
\]
where the implicit constant depends only on $s$.
For the smaller term in this minimum, the corresponding variance bound in \Cref{lem:moment-ht-relvar} gives
\[
  R_s(q)\le R_s(q/2)=O(D_s^{-s}\eps^2),
\]
using $D_s\ge1$ and $2s-1\ge s$.
Thus choosing $D_s$ sufficiently large ensures $R_s(q/2) \le \eps^2$.

Fix the history through the second call to \Cref{alg:endpoint-capture} and suppose $\mathcal E_0$, $\Gamma_J\le KM_s$, and $\mathcal E_1$ hold.
The values $q$ and $\widehat W$ are then fixed, and the final step draws $N\sim\operatorname{Pois}(q)$ and takes $N$ independent fair endpoint samples.
Applying \Cref{lem:moment-estimated-total-degree} with $U=\widehat W$ yields
\[
  \frac{|\E[Z_q(\widehat W)]-M_s|}{M_s}
  =O\!\left(\eta(q)\min\{1,\eps\sqrt q\}\right)
  =O(\kappa\eps),
\]
because $\max\{\eps,q^{-1/2}\}\min\{1,\eps\sqrt q\}=\eps$.
Decreasing $\kappa$ if necessary, the bias is at most $\eps M_s/2$.
The variance bound in the same lemma gives
\[
  \Var[Z_q(\widehat W)]
  =O(D_s^{-s}\eps^2M_s^2).
\]
Because the bias is at most $\eps M_s / 2$, any deviation $|Z_q(\widehat{W}) - M_s| > \eps M_s$ requires $|Z_q(\widehat{W}) - \E[Z_q(\widehat{W})]| > \eps M_s / 2$.
Applying Chebyshev's inequality and increasing $D_s$ if necessary yields
\[
  \Prob[|Z_q(\widehat{W}) - M_s| > \eps M_s]
  \le \frac{4\Var[Z_q(\widehat{W})]}{\eps^2 M_s^2}
  \le \frac{1}{64}.
\]

Unconditionally, failure requires at least one of four events to occur: $\mathcal E_0^c$, stopping with $\Gamma_J > K M_s$ on $\mathcal E_0$, $\mathcal E_1^c$, or failure of the final estimate on $\mathcal E_0 \cap \{\Gamma_J \le K M_s\} \cap \mathcal E_1$.
Averaging the conditional bounds shows that each event occurs with probability at most $1/64$. By a union bound, the algorithm succeeds with probability at least $1 - 4/64 = 15/16 > 2/3$.

\subparagraph*{Expected query complexity.}
Define
\[
  Q_s(\nu, \eps) := \min\left\{ \frac{\nu^{1-1/s}}{\eps^{2/s}}, \, \frac{\nu^{1-1/(2s-1)}}{\eps^{2/(2s-1)}} \right\}.
\]
Using $\E[\min\{X, Y\}] \le \min\{\E[X], \E[Y]\}$ and the definition of $q$ in \Cref{alg:moment-estimation},
\[
  \E[q] \le D_s \left( 1 + \min\left\{ \frac{\E[A_J]}{\eps^{2/s}}, \, \frac{\E[A_J^{2s/(2s-1)}]}{\eps^{2/(2s-1)}} \right\} \right)
  = O\left( 1 + \min\left\{ \frac{\Lambda_s}{\eps^{2/s}}, \, \frac{\Lambda_s^{2s/(2s-1)}}{\eps^{2/(2s-1)}} \right\} \right),
\]
where the equality follows from the moments of $A_J$ (\Cref{lem:moment-search}).
Since $\nu \ge 1$ and $\eps \le 1$ give $Q_s(\nu, \eps) \ge 1$, substituting $\Lambda_s \le 2(2\nu)^{1-1/s}$ and $\Lambda_s^{2s/(2s-1)} = O(\nu^{1 - 1/(2s-1)})$ establishes $\E[q] = O(Q_s(\nu, \eps))$.
The initial call to \Cref{alg:endpoint-capture} requires $O(\sqrt{\nu})$ expected queries, and the parameter search requires $O(\Lambda_s)$ expected queries, which are both $O(Q_s(\nu, \eps))$.

Conditional on $q$, the second call to \Cref{alg:endpoint-capture} with accuracy $\eta(q)$ has expected query complexity $O(\eta(q)^{-2} + \sqrt{\nu}/\eta(q)) = O\left( q + \min\left\{ \sqrt{\nu q}, \,\sqrt{\nu}/\eps \right\} \right)$, since $\eta(q)^{-2} \le q/\kappa^2$ and $\sqrt{\nu}/\eta(q) = \kappa^{-1}\min\{\sqrt{\nu q}, \sqrt{\nu}/\eps\}$.
If $\nu \eps^2 \le 1$, then $\nu \le Q_s(\nu, \eps)$ and $2\sqrt{\nu q} \le \nu + q$.
Otherwise, $\sqrt{\nu}/\eps \le Q_s(\nu, \eps)$ for all $s \ge 2$.
In either case, the expected cost is $O(q + Q_s(\nu, \eps))$.
Taking expectations and using $\E[q] = O(Q_s(\nu, \eps))$ yields $O(Q_s(\nu, \eps))$ expected queries for the second call to \Cref{alg:endpoint-capture}.
Finally, evaluating $Z_q(\widehat W)$ uses at most $2\E[q] = O(Q_s(\nu, \eps))$ queries in expectation.
Summing all four bounds establishes the overall expected query complexity $O(Q_s(\nu, \eps))$.
Because both calls to \Cref{alg:endpoint-capture} and the parameter search terminate almost surely, \Cref{alg:moment-estimation} also terminates almost surely, completing the proof of \Cref{thm:intro-moment}.

\section{Analysis of Stars and Isolated Edges}
\label{sec:moment-stars-detail}
In this section, we provide the complete proofs and detailed analyses for the stars and isolated edges lower bound construction presented in \Cref{sec:moment-stars}.

\subparagraph*{Construction of stars and isolated edges.}
Fix integers $s, k \ge 2$, and assume $k \eps^2 \ge \kappa_s$ and $0 < \eps \le \eps_0(s)$, where $\kappa_s > 0$ is a sufficiently large constant and $\eps_0(s) > 0$ is a sufficiently small constant, both depending only on $s$.
Recall the parameter definitions from \Cref{sec:moment-stars}:
\[
  a := \lceil \eps^{-2} \rceil, \qquad
  r := \lceil 32 \eps a \rceil, \qquad
  D := \left\lfloor (k / (4a))^{1/s} \right\rfloor, \qquad
  L := a D^s.
\]
The base graphs $G_0$ and $G_1$ consist of $a$ and $a+r$ copies of $K_{1,D}$, together with $L$ and $L-rD$ isolated edges, respectively, padded with isolated vertices to share a common vertex set $V$.
The distributions $\mathsf H_0$ and $\mathsf H_1$ are induced by uniform random vertex relabelings of $G_0$ and $G_1$.

\begin{lemma}
\label{lem:stars-and-edges-parameters}
The parameters and base graphs $G_0$ and $G_1$ as defined above satisfy:
\begin{enumerate}
  \item $D \ge 2$, $r \le a/2$, $rD \le L/2$, and $L = aD^s = \Theta(k)$;
  \item $\nu(G_0) = a + L \le k$ and $\nu(G_1) \le a + L \le k$;
  \item $m(G_0) = m(G_1) = a(D^s + D)$;
  \item $(1+\eps) M_s(G_0) < (1-\eps) M_s(G_1)$.
\end{enumerate}
\end{lemma}
\begin{proof}
Since $a \le 2\eps^{-2}$ for $\eps \le 1$, choosing $\kappa_s$ sufficiently large guarantees $D \ge 2$ and $a < k/4$.
Because $a \ge \eps^{-2}$ and $\eps \le 1$, we have $\eps a \ge 1/\eps \ge 1$, which implies $r \le 32\eps a + 1 \le 33\eps a$.
Choosing $\eps_0(s)$ sufficiently small also ensures $r \le a/2$ and $r D \le L/2$.

Every nontrivial component in either graph has matching number $1$.
Because $D^s \le k/(4a)$, we have $L \le k/4$, yielding the matching number bounds:
\[
  \nu(G_0) = a + L < k, \qquad \nu(G_1) = a + r + L - rD \le a + L < k.
\]
Since $D \ge \frac{1}{2} (k/(4a))^{1/s}$, we have $L = a D^s \ge k / 2^{s+2} = \Theta(k)$.

Both base graphs have the same total edge count $m = a D + L = a(D^s + D)$.
Their degree moments satisfy $M_s(G_0) = a(D^s + D) + 2L \le 4 a D^s$ and $M_s(G_1) - M_s(G_0) = r(D^s - D) \ge \frac{r D^s}{2} \ge 4 \eps M_s(G_0)$.
For $\eps < 1/2$, this gap implies
\[
  (1 - \eps) M_s(G_1) - (1 + \eps) M_s(G_0) \ge (2\eps - 4\eps^2) M_s(G_0) > 0.
\]
Therefore $M_s(G_0)$ and $M_s(G_1)$ are $(1 \pm \eps)$-separated.
\end{proof}

\subparagraph*{Indistinguishability of stars and isolated edges.}
Although $\mathsf{H}_0$ and $\mathsf{H}_1$ differ in their numbers of stars and isolated edges, all stars are isomorphic, as are all isolated edges.
To prove indistinguishability, we consider an augmented oracle that returns each sampled edge together with the labeled component containing it, which determines the degrees of all its vertices.
To simulate any endpoint access algorithm making at most $q$ queries, we can answer its $t$-th random edge query with the edge sampled in the $t$-th response of the augmented oracle and answer each degree query from the revealed component containing that endpoint. Thus by the data processing inequality it suffices to bound the total variation distance between the distributions of these responses.

To decouple the queries hitting distinct components, we analyze an experiment where the number of queries is drawn as $N \sim \operatorname{Pois}(\lambda)$.
We call a component observed if at least one query hits it, and write $C_\star$ and $C_1$ for the numbers of observed stars and observed isolated edges.
We first show that the response sequence can be simulated solely from $(C_\star, C_1)$.

\begin{lemma}
\label{lem:moment-stars-component-counts}
Let $N \sim \operatorname{Pois}(\lambda)$ for $\lambda > 0$, independent of the graph and query randomness.
Under either hypothesis $\mathsf H_i$, the response sequence of $N$ queries to the augmented oracle can be simulated from $(C_\star, C_1)$ by a randomized rule that does not depend on $i$. 
\end{lemma}

\begin{proof}
By Poisson splitting, the numbers of queries hitting individual components are mutually independent Poisson variables: $\operatorname{Pois}(\lambda D/m)$ for each star and $\operatorname{Pois}(\lambda/m)$ for each isolated edge.
Since $D$ and $m$ are identical under both hypotheses, these rates depend only on the component type.

Given $(C_\star, C_1)$, the response sequence is generated as follows:
\begin{itemize}
  \item Take $C_\star$ disjoint copies of $K_{1,D}$ and $C_1$ disjoint isolated edges, and assign their vertices distinct labels chosen uniformly at random from $V$. 
  \item For each observed star and isolated edge, draw the number of queries independently from $\operatorname{Pois}(\lambda D/m)$ and $\operatorname{Pois}(\lambda/m)$, respectively, conditioned on being positive.
  \item For each query hitting a star, choose one of its $D$ edges uniformly at random; for an isolated edge, take the edge.
  \item Uniformly interleave all queries across the observed components, returning each chosen edge and its labeled component.
\end{itemize}
This generates the joint distribution of the response sequence: by Poisson splitting, the query counts and their uniform order are those of the actual queries. 
Since the sampled graph $G$ is a uniformly random relabeling of $G_i$, the labels in the first step form a uniformly random injection into $V$ that is independent of the remaining steps.
Because the rule depends only on the counts $(C_\star, C_1)$ and the shared parameters $(\lambda, D, m, V)$, it is identical under both $\mathsf H_0$ and $\mathsf H_1$.
\end{proof}

By the data processing inequality and \Cref{lem:moment-stars-component-counts}, it suffices to bound the total variation distance between the joint distributions of $(C_\star, C_1)$.
We use the following bound for binomials with different numbers of trials.
\begin{lemma}
\label{lem:binomial-trial-count}
For any $0 < p \le 1/2$, integer $n \ge 1$, and integer $u$ with $|u| \le n/2$,
\[
  \TV\left(\operatorname{Bin}(n, p), \, \operatorname{Bin}(n+u, p)\right) \le  \abs{u} \sqrt{\frac{2p}{n}}.
\]
\end{lemma}

\begin{proof}
Write $P_n := \operatorname{Bin}(n,p)$ and first compare $P_n$ with $P_{n+1}$.
For every $x \in \{0, \ldots, n+1\}$, the likelihood ratio is
\[
  \frac{P_n(x)}{P_{n+1}(x)} = \frac{\binom{n}{x} p^x (1-p)^{n-x}}{\binom{n+1}{x} p^x (1-p)^{n+1-x}} = \frac{n + 1 - x}{(n + 1)(1 - p)}.
\]
Taking $X \sim P_{n+1}$, we have $\E[X] = (n+1)p$ and $\Var[X] = (n+1)p(1-p)$.
Applying Cauchy--Schwarz to the mean deviation gives
\begin{align*}
  \TV(P_n, P_{n+1})
  &= \frac{1}{2} \E\left[\left|\frac{P_n(X)}{P_{n+1}(X)} - 1\right|\right]
  = \frac{\E[|X - (n+1)p|]}{2(n+1)(1-p)} \\
  &\le \frac{\sqrt{\Var[X]}}{2(n+1)(1-p)}
  = \frac{1}{2} \sqrt{\frac{p}{(n+1)(1-p)}}
  \le \sqrt{\frac{p}{n}},
\end{align*}
where the final inequality uses $1 - p \ge 1/2$ and $n + 1 > n$.
For $u \ge 0$, the triangle inequality across adjacent trial counts yields
\[
  \TV(P_n, P_{n+u}) \le \sum_{i=0}^{u-1} \TV(P_{n+i}, P_{n+i+1}) \le \sum_{i=0}^{u-1} \sqrt{\frac{p}{n+i}} \le u \sqrt{\frac{p}{n}}.
\]
For $u < 0$, summing instead from trial count $n - |u|$ to $n$ similarly yields
\[
  \TV(P_{n-|u|}, P_n) \le |u| \sqrt{\frac{p}{n - |u|}} \le |u| \sqrt{\frac{2p}{n}},
\]
since $n - |u| \ge n/2$.
\end{proof}

We now establish that a total variation bound for $N\sim\operatorname{Pois}(4q)$ queries transfers to $q$ queries at the cost of a factor of $2$.
\begin{lemma}
\label{lem:poisson-to-fixed}
Fix an integer $q\ge1$.
Under hypothesis $i\in\{0,1\}$, let $Y_1, Y_2, \dots$ be a random sequence of responses, $\mathsf P_i$ be the distribution of $(Y_1,\dots,Y_q)$, and $\mathsf Q_i$ be the distribution of $(Y_1,\dots,Y_N)$ where $N\sim\operatorname{Pois}(4q)$ is independent of the sequence.
Then $\TV(\mathsf P_0,\mathsf P_1)\le2\TV(\mathsf Q_0,\mathsf Q_1)$.
\end{lemma}

\begin{proof}
Retain the first $q$ responses when $N\ge q$, and output a common failure symbol $\bot$ when $N < q$.
Because the distribution of $N$ is identical under both hypotheses, $\bot$ occurs with the same probability $\Prob[N < q]$ under $\mathsf Q_0$ and $\mathsf Q_1$.
The retained responses conditional on $N \ge q$ follow $\mathsf P_0$ and $\mathsf P_1$, respectively.
Applying the data processing inequality under this truncation map gives $\Prob[N\ge q]\,\TV(\mathsf P_0,\mathsf P_1) \le\TV(\mathsf Q_0,\mathsf Q_1)$.
Chebyshev's inequality gives $\Prob[N<q]\le4/(9q)<1/2$, so $\Prob[N\ge q] > 1/2$, proving the claim.
\end{proof}

\begin{lemma}
\label{lem:moment-stars-fixed-sample-tv}
Let $\mathsf P_i$ denote the joint distribution of the responses to $q$ queries to the augmented oracle under $\mathsf H_i$.
For a sufficiently small constant $c_s>0$, every integer $1\le q\le c_s aD^{s-1}$ satisfies $\TV(\mathsf P_0,\mathsf P_1)\le 1/4$.
\end{lemma}

\begin{proof}
Assume $c_s \le 1/8$, draw $N\sim\operatorname{Pois}(\lambda)$ with $\lambda := 4q \le 4c_s a D^{s-1}$, and write $\mathsf Q_i$ for the joint distribution of the resulting responses under $\mathsf H_i$.

Conditional on $G$, Poisson splitting makes the numbers of queries hitting individual components independent.
A star is observed with probability $p_\star := 1 - e^{-\lambda D/m}$, while an isolated edge is observed with probability $p_1 := 1 - e^{-\lambda/m}$.
Consequently, $C_\star$ and $C_1$ are independent binomial random variables: $\operatorname{Bin}(a, p_\star)$ and $\operatorname{Bin}(L, p_1)$ under $\mathsf{H}_0$, and $\operatorname{Bin}(a+r, p_\star)$ and $\operatorname{Bin}(L-rD, p_1)$ under $\mathsf{H}_1$.

By \Cref{lem:moment-stars-component-counts}, the response sequence is generated from $(C_\star, C_1)$ by a randomized rule common to both hypotheses.
Applying the data processing inequality and subadditivity of total variation distance on independent coordinates yields
\[
  \TV(\mathsf Q_0, \mathsf Q_1)\le \TV\left(\operatorname{Bin}(a, p_\star), \, \operatorname{Bin}(a+r, p_\star)\right) + \TV\left(\operatorname{Bin}(L, p_1), \, \operatorname{Bin}(L-rD, p_1)\right).
\]
Since $m = a(D^s + D) \ge a D^s$, these probabilities satisfy $p_\star \le \lambda D / m \le 4c_s$ and $p_1 \le \lambda / m \le 4c_s / D$.
Applying \Cref{lem:binomial-trial-count} with $p_\star, p_1 \le 4c_s \le 1/2$, $r \le a/2$, and $rD \le L/2$ gives
\[
  \TV(\mathsf Q_0, \mathsf Q_1)
  \le \sqrt{2} \left(r \sqrt{\frac{p_\star}{a}} + rD \sqrt{\frac{p_1}{L}}\right)
  \le 2\sqrt{2c_s} \, \frac{r}{\sqrt{a}} \left(1 + D^{-(s-1)/2}\right)
  = O(\sqrt{c_s}),
\]
where the second step uses $L = a D^s$ and the last step uses $r \le 33\eps a$, $\eps^2 a \le 2$, and $D \ge 2$.
Together with \Cref{lem:poisson-to-fixed}, this yields $\TV(\mathsf P_0, \mathsf P_1) \le 2\TV(\mathsf Q_0, \mathsf Q_1) = O(\sqrt{c_s}) \le 1/4$ for a sufficiently small $c_s$.
\end{proof}

Set $q := \lfloor c_s aD^{s-1}\rfloor$.
Since $a \ge \eps^{-2}$ and $D \ge 2$, choosing $\eps_0(s)$ sufficiently small after fixing $c_s$ ensures $q \ge 1$ and $q = \Theta(aD^{s-1})$.
Substituting $a = \Theta(\eps^{-2})$ and $aD^s = \Theta(k)$ from \Cref{lem:stars-and-edges-parameters} yields
\[
  a D^{s-1} = a \left(\frac{a D^s}{a}\right)^{(s-1)/s} = a^{1/s} (a D^s)^{1-1/s} = \Theta\left(\frac{k^{1-1/s}}{\eps^{2/s}}\right).
\]
By \Cref{lem:moment-stars-fixed-sample-tv}, this choice of $q$ ensures $\TV(\mathsf P_0, \mathsf P_1) \le 1/4$.
Applying \Cref{lem:estimation-to-testing} alongside the separation in \Cref{lem:stars-and-edges-parameters} establishes the worst case expected query lower bound $\Omega(q) = \Omega\left(k^{1-1/s}/\eps^{2/s}\right)$ for the regime $k\eps^2 \ge \kappa_s$.

\section{Analysis of Paired Centers}
\label{sec:moment-paired-detail}
In this section, we provide the complete proofs and detailed analyses for the paired centers lower bound construction presented in \Cref{sec:moment-paired}.

\subparagraph*{Construction of paired centers.}
Fix integers $s, k \ge 2$, and assume $k \eps^2 \le \kappa_s$ and $0 < \eps \le \eps_0(s)$, where $\kappa_s > 0$ is the constant from \Cref{sec:moment-stars-detail} and $\eps_0(s) > 0$ is a sufficiently small constant depending only on $s$.
Recall the parameter definitions from \Cref{sec:moment-paired}:
\[
  k' := 2 \left\lfloor k/2 \right\rfloor, \qquad
  D := \max\left\{2, \left\lceil (k'\eps^2)^{-1/(2s-1)} \right\rceil\right\}, \qquad
  N := \frac{k' D}{2}, \qquad
  \beta := \frac{\beta_s}{\sqrt{N}},
\]
where $\beta_s > 0$ is a sufficiently large constant.

We consider $k'$ center vertices partitioned into $k'/2$ disjoint pairs.
To each center pair $(u,v)$, we attach $D$ local gadgets, giving $N = k'D/2$ gadgets in total.
Each gadget introduces two auxiliary vertices $x,y$ and selects between two configurations:
\begin{description}
  \item[Shared] Insert edges $\{u,x\}$ and $\{v,x\}$, so $x$ has degree $2$ and $y$ is isolated.
  \item[Separate] Insert edges $\{u,x\}$ and $\{v,y\}$, so both auxiliary vertices have degree $1$.
\end{description}
In both configurations, one edge is incident to $u$ and one to $v$.
For $\beta \le 1/4$ (as verified in \Cref{lem:moment-paired-separation}), we define distributions $\mathsf H_0$ and $\mathsf H_1$ over graphs on the same fixed set of $k' + 2N$ labeled vertices: independently for each of the $N$ gadgets, $\mathsf H_0$ chooses \emph{shared} with probability $1/2-\beta$ and \emph{separate} otherwise, while $\mathsf H_1$ chooses \emph{shared} with probability $1/2+\beta$ and \emph{separate} otherwise.

\begin{lemma}
\label{lem:paired-centers-moments}
Every graph in either support satisfies:
\begin{enumerate}
  \item $\nu(G) = k' \le k$;
  \item $m(G) = k'D = 2N$, and every center vertex has degree $D$;
  \item $M_s(G) = k'D^s + 2N + (2^s-2)Z$, where $Z$ denotes the number of shared gadgets.
\end{enumerate}
Furthermore, $Z \sim \operatorname{Bin}(N, 1/2-\beta)$ under $\mathsf H_0$ and $Z \sim \operatorname{Bin}(N, 1/2+\beta)$ under $\mathsf H_1$, so $\E_{\mathsf H_1}[M_s] - \E_{\mathsf H_0}[M_s] = 2(2^s-2)\beta N$ and $\Var_{\mathsf H_i}[M_s] \le (2^s-2)^2 N/4$.
\end{lemma}
\begin{proof}
Because every edge is incident to at least one of the $k'$ center vertices, the centers form a vertex cover of size $k'$, so $\nu(G) \le k'$.
Conversely, since $D \ge 2$, for each center pair $(u,v)$ we can match $u$ using an edge from its first gadget and $v$ using an edge from its second gadget.
Because auxiliary vertices across distinct gadgets are disjoint, this produces a matching of size $k'$, so $\nu(G) = k' \le k$.
Under either configuration, each gadget contributes one incident edge to each of $u$ and $v$.
Summing over the $D$ gadgets per center pair, every center vertex has degree $D$, and the total edge count is $m = k'D = 2N$.

For the degree moment, the $k'$ centers contribute $k'D^s$ to $M_s$.
In the separate configuration, auxiliary vertices $x$ and $y$ each have degree $1$, contributing $1^s + 1^s = 2$.
In the shared configuration, $x$ has degree $2$ and $y$ is isolated (degree $0$), contributing $2^s + 0^s = 2^s$.
With $Z$ denoting the number of shared gadgets, the total degree moment is $M_s = k'D^s + 2(N-Z) + 2^s Z = k'D^s + 2N + (2^s-2)Z$.
Because each gadget independently chooses shared with probability $1/2-\beta$ under $\mathsf H_0$ and $1/2+\beta$ under $\mathsf H_1$, $Z$ follows the stated binomial distributions.
The difference in expected moments is $(2^s-2)(\E_{\mathsf H_1}[Z] - \E_{\mathsf H_0}[Z]) = 2(2^s-2)\beta N$, and $\Var_{\mathsf H_i}[M_s] = (2^s-2)^2 \Var_{\mathsf H_i}[Z] \le (2^s-2)^2 N/4$.
\end{proof}
\begin{lemma}
\label{lem:moment-paired-separation}
Choosing $\beta_s$ sufficiently large and then $\eps_0(s)$ sufficiently small ensures that $\beta\le1/4$, that each $\mathsf H_i$ places probability at least $99/100$ on a set $I_i$ of moment values, and that $I_0$ and $I_1$ are $(1\pm\eps)$-separated.
\end{lemma}
\begin{proof}
Write $w_s := 2^s - 2$.
Since $D\ge2$, the auxiliary vertices contribute at most $2^sN=2^{s-1}k'D\le k'D^s$ to $M_s$.
Moreover, $\eps k'D^s/\sqrt N=\eps\sqrt{2k'}\,D^{s-1/2}$.
If $\lceil(k'\eps^2)^{-1/(2s-1)}\rceil\ge2$, then $D\le2(k'\eps^2)^{-1/(2s-1)}$, so this quantity is at most $2^s$.
Otherwise $D=2$ and $k'\eps^2\ge1$, so this quantity equals $2^s(k'\eps^2)^{1/2}\le2^s\kappa_s^{1/2}$, because $k'\eps^2\le k\eps^2\le\kappa_s$.
Hence
\begin{equation}
  k'D^s\le M_s\le2k'D^s,
  \qquad
  \frac{\eps k'D^s}{\sqrt N}\le C_s:=2^s\max\{1,\kappa_s^{1/2}\}.
  \label{eq:moment-paired-scale}
\end{equation}

Under either hypothesis, $\Var[Z] \le N/4$. Thus Chebyshev's inequality guarantees
\[
  \Prob[|Z-\E[Z]| > 5\sqrt N] \le \frac{\Var[Z]}{25N} \le \frac{1}{100}.
\]
Let $I_i$ be the set of attainable values of $M_s$ within $5 w_s \sqrt N$ of $\E_{\mathsf H_i}[M_s]$, so $\Prob_{\mathsf H_i}[M_s \in I_i] \ge 99/100$ and $I_i\subseteq[k'D^s,2k'D^s]$ by \eqref{eq:moment-paired-scale}.
The two means differ by $2w_s\beta N = 2w_s\beta_s\sqrt N$.
For $a \in I_0$ and $b \in I_1$, it follows that $b - a \ge 2w_s(\beta_s - 5)\sqrt N$.
To guarantee $(1+\eps)a < (1-\eps)b$, it suffices that $b - a > \eps(a+b)$.
Since \eqref{eq:moment-paired-scale} gives $\eps(a+b) \le 4\eps k'D^s \le 4C_s\sqrt N$, we choose the constant $\beta_s$ to satisfy $\beta_s > 5 + \frac{2C_s}{w_s}$.
With this choice of $\beta_s$, $b - a \ge 2w_s(\beta_s - 5)\sqrt N > 4C_s\sqrt N \ge 4\eps k'D^s \ge \eps(a+b)$, so $(1+\eps)a < (1-\eps)b$.
Finally, as $D\ge(k'\eps^2)^{-1/(2s-1)}$ and $k'\ge2$, we have $N=k'D/2\ge\frac12\eps^{-2/(2s-1)}$, so reducing $\eps_0(s)$ after fixing $\beta_s$ ensures $\beta=\beta_s/\sqrt N\le1/4$.
\end{proof}

\subparagraph*{Indistinguishability of paired centers.}
Recall that every gadget contributes two edges under either configuration.
A uniform random edge query is therefore equivalent to selecting one of the $N$ gadgets uniformly at random, and then choosing one of the two edges contributed by that gadget uniformly at random.
To prove indistinguishability, we consider an augmented oracle that selects a gadget uniformly at random independently of previous queries and of the assigned configurations, and reveals its identity and its configuration (shared or separate).
To simulate any endpoint access algorithm making at most $q$ queries, we can answer its $t$-th random edge query with one of the two edges of the gadget in the $t$-th response of the augmented oracle, chosen uniformly at random. 
Each degree query can be answered from the revealed configurations, as centers have degree $D$ and returned auxiliary vertices have degree $2$ or $1$ according to the configuration of their gadget.
Thus by the data processing inequality it suffices to bound the total variation distance between the distributions of these responses.

\begin{lemma}
\label{lem:moment-paired-query-tv}
Let $\mathsf P_i$ denote the joint distribution of the responses to $q$ queries to the augmented oracle under $\mathsf H_i$.
For $0<\beta\le1/4$, $\TV(\mathsf P_0,\mathsf P_1) = O(\beta\sqrt q)$.
In particular, for $\beta = \beta_s/\sqrt{N}$ (which satisfies $\beta\le1/4$ by \Cref{lem:moment-paired-separation}) and a sufficiently small constant $c_s > 0$, every integer $1 \le q \le c_s N$ satisfies $\TV(\mathsf P_0, \mathsf P_1) \le 1/4$.
\end{lemma}

\begin{proof}
The responses are determined by the sequence of selected gadgets (which has the same distribution under both hypotheses) together with the configurations of the distinct selected gadgets.
Conditional on this sequence, the indicators of shared configurations at the distinct selected gadgets are independent Bernoulli variables with parameter $1/2-\beta$ under $\mathsf H_0$ and $1/2+\beta$ under $\mathsf H_1$.
For $\beta\le1/4$, the inequality $\log x\le x-1$ gives
\[
  \KL\left(\operatorname{Ber}(1/2-\beta)\middle\|\operatorname{Ber}(1/2+\beta)\right)
  \le\frac{4\beta^2}{(1/2+\beta)(1/2-\beta)}
  =O(\beta^2).
\]
By the chain rule, the sequence of selected gadgets contributes zero divergence and each of the at most $q$ distinct selected gadgets contributes the divergence above, so $\KL(\mathsf P_0\|\mathsf P_1)=O(q\beta^2)$.
Pinsker's inequality then yields $\TV(\mathsf P_0,\mathsf P_1)=O(\beta\sqrt q)$.
For $\beta = \beta_s/\sqrt{N}$ and $q \le c_s N$, this bound becomes $O(\beta_s\sqrt{c_s})$, which is at most $1/4$ for a sufficiently small constant $c_s > 0$.
\end{proof}
Set $q := \lfloor c_s N \rfloor$.
Since $N\ge\frac12\eps^{-2/(2s-1)}$ by the proof of \Cref{lem:moment-paired-separation}, choosing $\eps_0(s)$ sufficiently small after fixing $c_s$ ensures $q \ge 1$ and $q = \Theta(N)$.
Moreover, $(k'\eps^2)^{-1/(2s-1)}\le D\le2\max\{1,\kappa_s^{1/(2s-1)}\}(k'\eps^2)^{-1/(2s-1)}$, where the upper bound uses $k'\eps^2\le\kappa_s$ in the case $k'\eps^2\ge1$, in which $D=2$.
Substituting these bounds and $k/2\le k'\le k$ into $N = k'D/2$ yields
\[
  N = \Theta\left(k(k\eps^2)^{-1/(2s-1)}\right) = \Theta\left(\frac{k^{1 - 1/(2s-1)}}{\eps^{2/(2s-1)}}\right).
\]
By \Cref{lem:moment-paired-query-tv}, this choice of $q$ ensures $\TV(\mathsf P_0, \mathsf P_1) \le 1/4$.
Applying \Cref{lem:estimation-to-testing} alongside the separation in \Cref{lem:moment-paired-separation} establishes the worst case expected query lower bound $\Omega(q) = \Omega\left(k^{1 - 1/(2s-1)}/\eps^{2/(2s-1)}\right)$ for the regime $k\eps^2 \le \kappa_s$, completing the proof of \Cref{thm:intro-moment-lb}.

\section{Weighted Sampling, a Prior Lower Bound, and Instance Sensitivity}
\label{sec:proportional-sampling}
\Cref{alg:endpoint-capture,alg:moment-estimation} naturally generalize from graph degrees to arbitrary nonnegative weights under the weighted sampling model studied in~\cite{MotwaniPanigrahyXu2007, BerettaTetek2024,BhattacharyaPradhan2025}.
In this section, we establish this transfer and contrast our bounds with prior work.
We then revisit the lower bound stated by Bhattacharya and Pradhan~\cite{BhattacharyaPradhan2025}, show that it does not hold uniformly over the stated parameter range, and examine the construction used in its proof.
Finally, we present an instance sensitive analysis showing that our moment estimator automatically adapts to favorable instance structure, using substantially fewer samples than its worst case bound without any advice.

\subsection{Total Weight and Moment Estimation}
\label{sec:proportional-transfer}
\subparagraph*{The weighted sampling model.}
Let $A$ be a set of $n$ elements, where each element $a \in A$ has an unknown nonnegative weight $w(a)$ with $W := \sum_{a \in A} w(a) > 0$ denoting the total weight.
In the weighted sampling model of Beretta and T\v{e}tek~\cite{BerettaTetek2024} and Bhattacharya and Pradhan~\cite{BhattacharyaPradhan2025}, an algorithm accesses $A$ via an oracle that independently returns an element $a \in A$ with probability proportional to its weight,
$p(a) := w(a)/W$.
The oracle returns both a stable identifier of $a$ (allowing repeat observations to be recognized) and its weight $w(a)$.
As discussed in \Cref{sec:model}, endpoint access on graphs directly corresponds to this model: setting $A = V$, element weights as vertex degrees $w(v) = d_v$, and total weight $W = 2m$, one $\RE()$ query followed by a fair coin flip draws an endpoint with probability $p_v = d_v/(2m)$, and a subsequent $\Deg(v)$ query reveals its weight.
Because endpoint access on graphs uses at most two queries to draw a fair endpoint sample and reveal its degree, any algorithm designed for the weighted sampling model transfers to endpoint access with at most a factor of two query overhead.
Two estimation problems have been studied under this weighted sampling model:
\begin{enumerate}
  \item \textbf{Total weight estimation:} Beretta and T\v{e}tek~\cite{BerettaTetek2024} studied $(1\pm\eps)$-estimation of the total weight $W = \sum_{a \in A} w(a)$.
  When $n$ is known, their algorithm uses $O(\sqrt{n}/\eps)$ samples; without advice, it uses $O(\log n / \eps^2 + \sqrt{n}/\eps)$ samples in expectation.
  \item \textbf{Moment estimation:} Bhattacharya and Pradhan~\cite{BhattacharyaPradhan2025} investigated $(1\pm\eps)$-estimation of the $s$-th moment $M_s := \sum_{a \in A} w(a)^s$ for every real $s > 0$; we compare with their results for integers $s \ge 2$.
  When $n$ is known, their algorithm first estimates $W$ via~\cite{BerettaTetek2024} and then draws $O(n^{1-1/s}/\eps^2)$ additional weighted samples, using $O(n^{1-1/s}/\eps^2)$ samples in total.
\end{enumerate}
These two targets correspond to estimating the edge count $m$ and the degree moment $M_s$ in our graph setting.

\subparagraph*{Improvements in the weighted sampling model.}
While \Cref{alg:endpoint-capture,alg:moment-estimation} are designed for graph parameter estimation, they transfer to the weighted sampling model: each fair endpoint sample is replaced by a weighted sample, and returned weights replace degree queries.
Since \Cref{lem:capture-return-law} does not use the graph structure, the accuracy and ratio moment bounds of \Cref{thm:endpoint-capture} carry over unchanged.
Applying the Cauchy--Schwarz argument of \Cref{lem:captured-endpoint-mass} with the set of positive weight elements in place of $C^\star$ yields $\E[p_S] \ge q/(n+q)$, so the expected number of samples is $O(1/\eps^2 + \sqrt{n}/\eps)$.
For moment estimation, the analysis in \Cref{sec:moment-detail} transfers to arbitrary nonnegative weights: since H\"older's inequality ensures $\Lambda_s \le n^{1-1/s}$, \Cref{alg:moment-estimation} achieves the full expected sample complexity of \Cref{thm:intro-moment} with $n$ in place of $\nu$.
In the comparisons that follow, we focus on the simpler consequence $O(n^{1-1/s}/\eps^{2/s})$, which already provides a polynomial improvement in the dependence on $\eps$. 
To summarize, \Cref{alg:endpoint-capture,alg:moment-estimation} give the following results:
\begin{itemize}
  \item For total weight estimation, \Cref{alg:endpoint-capture} is advice-free, and doubling its output gives a $(1\pm\eps)$-estimate of $W$ with probability at least $63/64$ using $O(1/\eps^2 + \sqrt{n}/\eps)$ expected samples.
  This eliminates the $\log n$ factor from the bound $O(\log n / \eps^2 + \sqrt{n}/\eps)$ of Beretta and T\v{e}tek~\cite{BerettaTetek2024} without advice.
  \item For moment estimation, \Cref{alg:moment-estimation} is advice-free and returns a $(1\pm\eps)$-estimate of $M_s$ with probability at least $15/16$ using $O(n^{1-1/s}/\eps^{2/s})$ expected samples for every fixed integer $s \ge 2$.
  Compared with the $O(n^{1-1/s}/\eps^2)$ bound of Bhattacharya and Pradhan~\cite{BhattacharyaPradhan2025}, our algorithm removes the need for $n$ as advice and polynomially improves the dependence on $\eps$, replacing $\eps^{-2}$ with $\eps^{-2/s}$.
  Furthermore, we discuss the instance sensitive behavior of \Cref{alg:moment-estimation} in \Cref{sec:moment-instance-sensitive}.
\end{itemize}

\subsection{Comparison with a Prior Lower Bound}
\label{sec:bp25-lower-bound}
Bhattacharya and Pradhan~\cite{BhattacharyaPradhan2025} state that for every $s > 1$ and $\eps \in (0,1)$, obtaining a $(1\pm\eps)$-estimate of the moment $M_s = \sum_{a \in A} w(a)^s$ of an $n$-element weighted set at constant failure probability requires $\Omega(n^{1-1/s}/\eps^{2})$ weighted samples.
Since \Cref{alg:moment-estimation} achieves a polynomially better dependence on $\eps$ (\Cref{sec:proportional-transfer}) for every fixed integer $s \ge 2$, we revisit this lower bound statement.
This comparison does not depend on whether samples are counted in expectation or in the worst case: when $n$ is known, Markov's inequality establishes that stopping \Cref{alg:moment-estimation} after $16$ times its expected sample bound lowers its success probability by at most $1/16$.
We first give an example with $s = 2$ and accuracy parameter $\eps = 1/n$ showing that the lower bound does not hold uniformly over the stated parameter range, and then discuss the corresponding construction in~\cite{BhattacharyaPradhan2025}.
This discussion does not affect their algorithm, which uses $O(n^{1-1/s}/\eps^2)$ weighted samples when $n$ is known.

\subparagraph*{Example with small accuracy parameter.}
To illustrate a setting within the stated parameter range where this lower bound does not hold, consider second moment estimation ($s = 2$) with known $n \ge 2$ and accuracy parameter $\eps = 1/n$.
For this setting, the stated lower bound evaluates to $\Omega(n^{5/2})$ weighted samples at constant failure probability.
We exhibit an estimator that uses only $O(n^{3/2}\log n) = o(n^{5/2})$ samples and succeeds with probability at least $2/3$ on every weighted set $A$ of $n$ elements.

The estimator draws $q := \lceil n^{3/2}\log(3n)\rceil$ independent weighted samples and returns $\widehat M_2 := \sum_{a \in S} w(a)^2$, where $S \subseteq A$ is the set of distinct observed elements. 
Call an element $a$ heavy if $w(a)^2 \ge M_2/n^2$.
To obtain a $(1\pm\eps)$-estimate, it suffices to observe every heavy element, as the remaining elements contribute less than $n(M_2/n^2) = M_2/n = \eps M_2$ in total. 
By Cauchy--Schwarz, $W \le \sqrt{n M_2}$, so every heavy element has sampling probability
\[
  p(a) = \frac{w(a)}{W} \ge \frac{\sqrt{M_2}/n}{\sqrt{n M_2}} = n^{-3/2}.
\] 
Therefore, the probability that a given heavy element remains unobserved after $q$ samples is at most $(1 - p(a))^q \le e^{-q n^{-3/2}} \le 1/(3n)$.
By a union bound over at most $n$ heavy elements, every heavy element is observed with probability at least $2/3$.
On this event, the unobserved contribution satisfies
\[
  0 \le M_2 - \widehat M_2 = \sum_{a \in A \setminus S} w(a)^2 < \eps M_2,
\]
so $\widehat M_2$ is a $(1\pm\eps)$-estimate of $M_2$.
More generally, applying H\"older's inequality in place of Cauchy--Schwarz gives an estimator using $O(n^{1+1/s}\log n) = o(n^{3-1/s})$ samples for any fixed integer $s \ge 2$ with $\eps = 1/n$, whereas the stated lower bound evaluates to $\Omega(n^{3-1/s})$.

\subparagraph*{Instance construction and integer cardinalities.}
The lower bound proof in~\cite{BhattacharyaPradhan2025} uses a pair of weighted sets designed to be difficult to distinguish with weighted samples.
The construction defines two sets $A_1$ and $A_2$ of $n_1 + n_2$ elements, in each of which $n_1$ elements have weight $d_1 := n^{1-1/s}\eps^{1/(s-1)}$.
In $A_1$, the remaining $n_2$ elements have weight zero; in $A_2$, they have weight $d_2 := n$, where
\[
  n_2 = \frac{n\eps^{(2s-1)/(s-1)}}{n + \eps^{(2s-1)/(s-1)}},
  \qquad
  n_1 = n - n_2.
\]
For this construction to represent valid finite sets, $n_1$ and $n_2$ must be nonnegative integers.
However, for every positive integer $n$, $s > 1$, and $\eps \in (0, 1)$, the exponent satisfies $(2s-1)/(s-1) > 0$, so $0 < n_2 < \eps^{(2s-1)/(s-1)} < 1$.
Because $n_2$ lies strictly between $0$ and $1$, this parameter specification does not yield integer cardinalities for any $\eps \in (0,1)$, so the two instances cannot be realized as finite weighted sets, including in the setting of the example above.

\subparagraph*{Insufficient moment separation.}
A reduction from estimation to testing requires the moments of the two instances to be $(1\pm\eps)$-separated, that is, no value may be a $(1\pm\eps)$-estimate of both.
In~\cite{BhattacharyaPradhan2025}, the two instances are chosen so that $M_s(A_2) = (1+\eps) M_s(A_1)$.
However, a ratio of $1+\eps$ does not give $(1\pm\eps)$-separation, since
\[
  (1-\eps) M_s(A_2)
  = (1-\eps^2) M_s(A_1)
  \le M_s(A_1)
  \le (1+\eps) M_s(A_2).
\]
Thus $M_s(A_1)$ is a $(1\pm\eps)$-estimate of both $M_s(A_1)$ and $M_s(A_2)$, so a correct estimator need not distinguish $A_1$ from $A_2$.
The two moments are $(1\pm\eps)$-separated when $M_s(A_2)/M_s(A_1) > (1+\eps)/(1-\eps)$.
Running the construction with $3\eps$ in place of $\eps$ gives the ratio $1+3\eps$, which exceeds $(1+\eps)/(1-\eps)$ when $\eps < 1/3$, but then $n_2$ again lies strictly between $0$ and $1$.

\subsection{Instance Sensitive Analysis}
\label{sec:moment-instance-sensitive}
Fix an integer $s \ge 2$ and recall $W = \sum_{v \in V} d_v > 0$ and $M_s = \sum_{v \in V} (d_v)^s$.
In addition to their worst case bound in terms of $n$, Bhattacharya and Pradhan~\cite{BhattacharyaPradhan2025} state an instance sensitive sample bound of $O(\sqrt{n}/\eps + \rho / \eps^2)$ parameterized by their moment-density parameter
\begin{equation}
  \rho
  := \max_{L \subseteq V:\, \sum_{v \in L} d_v > 0}
  \frac{\sum_{v \in L} (d_v)^s / M_s}{\sum_{v \in L} d_v / W}
  = \frac{W(\max_v d_v)^{s-1}}{M_s}.
  \label{eq:bp25-degree-rho}
\end{equation}
The identity in~\eqref{eq:bp25-degree-rho} holds because the subset ratio is $W/M_s$ times a degree weighted average of $(d_v)^{s-1}$, which is maximized by a singleton vertex of maximum degree.
Thus $\rho$ measures the maximum ratio between a vertex's share of the moment and its sampling probability.

The published proof bounds the normalized second moment $W\sum_v (d_v)^{2s-1}/M_s^2$ by $\rho$, and the worst case analysis also bounds this ratio by $W/M_s^{1/s}$ before relaxing it to $n^{1-1/s}$.
However, the algorithm proposed in~\cite{BhattacharyaPradhan2025} fixes $\Theta(n^{1-1/s}/\eps^2)$ samples on every input, and the paper does not specify how to choose a count of order $\rho/\eps^2$ without knowing $\rho$ or an upper bound on it.
Because this stated bound is not achieved by the proposed algorithm, we do not compare against it as an algorithmic baseline.

By contrast, while the formal guarantee of \Cref{thm:intro-moment} is stated in terms of the matching number $\nu$ to provide a rigorous worst case bound, \Cref{alg:moment-estimation} adapts its expected sample size without receiving $\Lambda_s$, $\nu$, or $n$ as input.
Specifically, \Cref{lem:moment-search} bounds the expected search cost by $O(\Lambda_s)$ and establishes moment bounds on the selected $A_J$ used to determine the final expected sample size $q$.
Substituting these bounds into the definition of $q$ in \Cref{alg:moment-estimation} yields
\begin{equation}
  \E[q] = O\!\left(
    \min\left\{
      \frac{\Lambda_s}{\eps^{2/s}},
      \frac{\Lambda_s^{2s/(2s-1)}}{\eps^{2/(2s-1)}}
    \right\}
  \right).
  \label{eq:moment-instance-sample-size}
\end{equation}
The following family demonstrates that this bound can be substantially smaller than the worst case bound in terms of $\nu$.

\subparagraph*{A star and a matching.}
For an integer $k \ge 2$, let $G_k$ be the disjoint union of a star with $k^2$ leaves and a matching of $k$ edges.
Then,
\[
  n = k^2 + 2k + 1, \qquad
  \nu(G_k) = k + 1, \qquad
  W = 2(k^2 + k), \qquad
  M_s = k^{2s} + k^2 + 2k.
\]
Consequently, for every fixed integer $s \ge 2$,
\[
  1 \le \Lambda_s
  = \frac{2(k^2 + k)}{(k^{2s} + k^2 + 2k)^{1/s}}
  \le 4.
\]
The star center contributes almost all of the moment and carries endpoint mass $p_c = k^2 / [2(k^2 + k)] \ge 1/4$.
This mass concentration also enables \Cref{alg:endpoint-capture} to estimate $W$ quickly.

To see this, consider a call to \Cref{alg:endpoint-capture} with accuracy $0 < \eta \le 1$ and threshold $R = \lceil 64\eta^{-2} \rceil$.
Fix the history at the start of any reached iteration with batch size $b \ge 1$.
Recall that $H$ counts the hits in the first batch's observed set $S$ among the second batch.
Because the first batch draws $b$ fair endpoint samples, the center vertex is included in $S$ with probability $1 - (1 - p_c)^b \ge p_c$.
Retaining only the center's contribution to $p_S$ yields $\E[p_S] \ge p_c^2 \ge 1/16$, so $\E[H] = b \E[p_S] \ge b p_c^2 \ge b/16$.
The variance bound $\Var[H] \le 2\E[H]$ from \Cref{lem:captured-endpoint-mass} remains valid.
Hence every reached iteration with $b \ge 128R$ has $\E[H] \ge 8R$ and, by Chebyshev's inequality as in \Cref{lem:capture-hit-threshold}, continues with probability at most $1/4$.
Because \Cref{lem:capture-fixed-scale-cost} bounds the expected cost of each reached iteration by $4b$, summing over the iterations gives an expected query complexity of $O(R) = O(\eta^{-2})$ for the entire call, uniformly in $k$.

The algorithm requires no knowledge of the star center or the graph structure for this bound to hold.
The full query complexity of \Cref{alg:moment-estimation} on $G_k$ is therefore accounted for as follows:
\begin{itemize}
  \item The initial estimation of $W$ with accuracy $1/4$ costs $O(1)$ expected queries as above.
  \item The parameter search costs $O(\Lambda_s) = O(1)$ expected queries because $\Lambda_s \le 4$.
  \item The final expected sample size satisfies $\E[q] = O(\eps^{-2/(2s-1)})$ by~\eqref{eq:moment-instance-sample-size}.
  \item The final estimation of $W$ with accuracy $\eta(q) = \kappa \max\{\eps, q^{-1/2}\}$ uses $O(\E[\eta(q)^{-2}]) = O(\E[\min\{\eps^{-2}, q\}]) = O(\E[q])$ queries in expectation.
  \item The final Horvitz--Thompson moment estimation requires at most $2\E[q]$ queries in expectation.
\end{itemize}
Summing these contributions, the total expected query complexity of \Cref{alg:moment-estimation} on $G_k$ is $O(\eps^{-2/(2s-1)})$, independent of $k$.

By contrast, for fixed $s \ge 2$ and $\eps \in (0, 1)$, the bound $O(n^{1-1/s}/\eps^{2/s})$ evaluates to $\Theta(k^{2-2/s})$ on this family.
The bound of \Cref{thm:intro-moment} in terms of $\nu(G_k)$ evaluates to $\Theta(k^{1-1/s})$, whereas the expected query complexity of \Cref{alg:moment-estimation} on $G_k$ does not grow with $k$.
Thus, although the matching number gives a tighter worst case guarantee than the vertex count, \Cref{alg:moment-estimation} exploits further degree concentration and uses substantially fewer queries than this guarantee on such instances without any advice.
Identifying further structural parameters that tightly characterize this instance sensitive complexity remains an interesting direction for future work.

\end{document}